\documentclass[letterpaper,11pt]{article}
 
\usepackage{microtype}

\title{Limits on All Known (and Some Unknown) \\Approaches to Matrix Multiplication}
\author{Josh Alman\footnote{MIT CSAIL and EECS, jalman@mit.edu. Supported by two NSF Career Awards.} \and Virginia Vassilevska Williams\footnote{MIT CSAIL and EECS, virgi@mit.edu. Partially supported by an NSF Career Award, a Sloan Fellowship, NSF Grants CCF-1417238, CCF-1528078 and CCF-1514339, and BSF Grant BSF:2012338.}}

\usepackage[utf8]{inputenc}
\usepackage{amsmath,amsfonts,amssymb,amsthm,fullpage}
\usepackage{graphicx}
\usepackage{caption}
\usepackage{subcaption}
\usepackage{enumerate}
\usepackage{url}

\newtheorem{theorem}{Theorem}[section]

\newtheorem{corollary}{Corollary}[section]
\newtheorem{lemma}{Lemma}[section]
\newtheorem{claim}{Claim}[section]
\newtheorem{definition}{Definition}[section]

\theoremstyle{remark}
\newtheorem{remark}{Remark}[section]
\newtheorem{example}{Example}[section]

\def \C {{\mathbb C}}
\def \Z {{\mathbb Z}}
\def \N {{\mathbb N}}
\def \F {{\mathbb F}}
\def \R {{\mathbb R}}
\def \E {{\mathbb E}}

\def\eps{\varepsilon}

\begin{document}
\date{}
\maketitle

\begin{abstract}
We study the known techniques for designing Matrix Multiplication algorithms. The two main approaches are the Laser method of Strassen, and the Group theoretic approach of Cohn and Umans. We define a generalization based on zeroing outs which subsumes these two approaches, which we call the Solar method, and an even more general method based on monomial degenerations, which we call the Galactic method.

We then design a suite of techniques for proving lower bounds on the value of $\omega$, the exponent of matrix multiplication, which can be achieved by algorithms using many tensors $T$ and the Galactic method. Some of our techniques exploit `local' properties of $T$, like finding a sub-tensor of $T$ which is so `weak' that $T$ itself couldn't be used to achieve a good bound on $\omega$, while others exploit `global' properties, like $T$ being a monomial degeneration of the structural tensor of a group algebra.

Our main result is that there is a universal constant $\ell>2$ such that a large class of tensors generalizing the Coppersmith-Winograd tensor $CW_q$ cannot be used within the Galactic method to show a bound on $\omega$ better than $\ell$, for any $q$. We give evidence that previous lower-bounding techniques were not strong enough to show this. We also prove a number of complementary results along the way, including that for any group $G$, the structural tensor of $\C[G]$ can be used to recover the best bound on $\omega$ which the Coppersmith-Winograd approach gets using $CW_{|G|-2}$ as long as the asymptotic rank of the structural tensor is not too large.
\end{abstract}

\thispagestyle{empty}
\newpage
\setcounter{page}{1}

\section{Introduction}
A fundamental problem in theoretical computer science is to determine the time complexity of Matrix Multiplication (MM), one of the most basic linear algebraic operations. The question typically translates to determining the {\em exponent of matrix multiplication}: the smallest real number $\omega$ such that the product of two $n\times n$ matrices over a field $\F$ can be determined using $n^{\omega+o(1)}$ operations over $\F$. Trivially, $2\leq \omega\leq 3$.
Many have conjectured over the years that $\omega=2$. This conjecture is extremely attractive:  a near-linear time algorithm for MM would immediately imply near-optimal algorithms for many problems. 

Almost $50$ years have passed since Strassen~\cite{strassen} first showed that $\omega\leq 2.81<3$. Since then, an impressive toolbox of techniques has been developed to obtain faster MM algorithms, culminating in the current best bound $\omega<2.373$ \cite{legall,v12}.
Unfortunately, this bound is far from $2$, and the current methods seem to have reached a standstill. Recent research has turned to proving limitations on the two main MM techniques: the Laser method of Strassen~\cite{laser} and the Group theoretic method of Cohn and Umans~\cite{cohn2003group}.

Both Coppersmith and Winograd~\cite{coppersmith} and Cohn et al.~\cite{cohn2005group} proposed conjectures which, if true, would imply that $\omega=2$. The first conjecture works in conjunction with the Laser method, and the second with the Group-theoretic method.
The first ``technique limitation'' result was by Alon, Shpilka and Umans~\cite{sunfl} who showed that both conjectures would contradict the widely believed Sunflower conjecture of Erd\"os and Rado.

Ambainis, Filmus and Le Gall~\cite{ambainis} formalized the specific implementation of the Laser method proposed by Coppersmith and Winograd~\cite{coppersmith} which is used in the recent papers on MM. They gave limitations of this implementation, and in particular showed that the exact approach used in \cite{coppersmith,stothers,legall,v12} cannot achieve a bound on $\omega$ better than $2.3078$. The analyzed approach, the ``Laser Method with Merging'', is a bit more general than the approaches in \cite{coppersmith,stothers,legall,v12}: in a sense it corresponds to a dream implementation of the exact approach.

Blasiak et al.~\cite{blasiak} considered the group theoretic framework for developing MM algorithms proposed by Cohn and Umans~\cite{cohn2003group}, and showed that this approach cannot prove $\omega=2$ using any fixed abelian group. In follow-up work, Sawin \cite{sawin2017bounds} extended this to any fixed non-abelian group, and Blasiak et al. \cite{blasiak2017groups} extended it to a host of families of non-abelian groups.

Alman and Vassilevska W.~\cite{almanitcs} considered a generalization of the Laser method and proved limitations on this generalization when it is applied to any tensor which is a {\em monomial degeneration} of the structure tensor of the group algebra $\C[C_q]$ of the cyclic group $C_q$ of order $q$. (See Section~\ref{sec:prelims} for the definitions.) 
The bounds on $\omega$ achieved by known implementations of the Laser method \cite{laser,coppersmith,stothers,legall,v12} can all be obtained from tensors of this form. 
The formalization also subsumes the group theoretic approach  applied to $C_q$. The main result of~\cite{almanitcs} is that this generalized approach cannot achieve $\omega=2$ for any fixed $q$. 

All limitations proven so far suffer from several weaknesses:

\begin{itemize}
\item 
All three of \cite{blasiak}, \cite{blasiak2017groups} and \cite{almanitcs} show how {\em some} approach that can yield the current best bounds on $\omega$ cannot give $\omega=2$. None of the three works actually prove that one cannot use the particular tensor $CW_q$ used in recent work~\cite{coppersmith,stothers,v12,legall} to show $\omega=2$. \cite{almanitcs} proved this limitation for a rotated version of $CW_q$, but only for small $q$. 
Although \cite{blasiak} and \cite{blasiak2017groups} do not say which version their proofs apply to, in this paper we give evidence that $CW_q$ does not embed easily in a group tensor, and so it is likely that their proofs could also only apply to a rotated version of $CW_q$, and not to $CW_q$ itself. Moreover, even for the Coppersmith-Winograd-like tensors for which the known limitations do apply, it is only shown that for a {\em fixed} $q$ one cannot derive $\omega=2$. In particular, so far the lower bounds $\omega_q$ on what $\omega$ one can achieve for a value $q$ approached $2$. This left open the possibility to prove $\omega=2$ by analyzing $CW_q$ in the limit as $q\rightarrow\infty$.

\item All limitations proven so far are for very specific attacks on proving $\omega=2$. While the proofs of \cite{ambainis} apply directly to $CW_q$, they only apply to the restricted Laser Method with Merging, and no longer apply to slight changes to this. The proofs in \cite{blasiak} and \cite{blasiak2017groups} are tailored to the group theoretic approach and do not apply (for instance) to the Laser method on ``non-group'' tensors. While the limits in \cite{almanitcs} do apply to a more general method than both the group theoretic approach and the Laser method, they only work for specific types of tensors, which in particular do not include $CW_q$.
\end{itemize}

\paragraph{Our results.}
All known approaches to matrix multiplication follow the following outline. First, obtaining a bound on $\omega$ corresponds to determining the {\em asymptotic rank} of the matrix multiplication tensor $\langle N,N,N\rangle$ (see the Preliminaries for a formal definition). Because getting a handle on this asymptotic rank seems difficult, one typically works with a tensor $t$ (or a tensor family) whose asymptotic rank $r$ is known. Then, to analyze the asymptotic rank of matrix multiplication, one considers large tensor powers $t^{\otimes n}$ of $t$ and attempts to ``embed'' $\langle N, N,N\rangle$ into $t^{\otimes n}$ for large $N$ without increasing the asymptotic rank. In effect, one is showing that the recursive $O(r^n)$ time algorithm for computing $t^{\otimes n}$ can be used to multiply $N\times N$ matrices. This gives a bound on $\omega$ from $N^\omega\leq r^n$. The larger $N$ is in terms of $n$, the smaller the bound on $\omega$.

When embedding matrix multiplication into a tensor power $t^{\otimes n}$, we would like the embedding to have the property that if $a$ embeds in $b$, then the asymptotic rank of $a$ is upper bounded by the asymptotic rank of $b$. This way, our embedding gives an upper bound on the asymptotic rank of matrix multiplication, and hence on $\omega$. The most general type of embedding that preserves asymptotic rank in this way is a so called {\em degeneration} of the tensor $t^{\otimes n}$. A more restricted type of rank-preserving embedding is a so called {\em monomial degeneration}. The embeddings used in all known approaches for upper bounding $\omega$ so far are even more restricted {\em zeroing outs}. The laser method is a restricted type of zeroing out that has only been applied so far to tensors that look like matrix multiplication tensors or to ones related to the Coppersmith-Winograd tensor. The group theoretic approach gives clean definitions that imply the existence of a zeroing out of a group tensor into a matrix multiplication tensor.
(See the preliminaries for formal definitions.)

We define three very general methods of analyzing tensors. There are {\em no known} techniques to analyze tensors in this generality.
\begin{itemize}
\item The {\bf Solar} Method applied to a tensor $t$ of asymptotic rank $r$ considers $t^{\otimes n}$ for large $n$, then considers all possible ways to {\em zero out} $t^{\otimes n}$ into a disjoint sum $\langle a_1, b_1, c_1 \rangle \oplus \cdots \oplus \langle a_m, b_m, c_m \rangle$ of matrix multiplication tensors, giving a bound on $\omega$ from the asymptotic sum inequality of $\sum_{i=1}^m (a_i b_i c_i)^{\omega/3} \leq r^n$, and then takes the minimum (or $\liminf$) of all bounds on $\omega$ which can be achieved in this way. This method already subsumes both the group theoretic method and the laser method. It is also much more general, as it is unclear whether the two known techniques produce the best possible zeroing outs even for specific tensors. 

\item The {\bf Galactic} Method replaces the zeroing out in the Solar Method with more powerful {\em monomial degenerations}. Since monomial degenerations are strictly more powerful than zeroing outs in general, this leads to even more possible embeddings of disjoint sums of matrix multiplication tensors.

\item The {\bf Universal} Method again replaces the monomial degenerations of the Galactic Method with the even more powerful {\em degenerations}.
\end{itemize}

We note that the methods only differ when they are applied to the same tensor $t$. Trivially, any one of the methods can find the best bound on $\omega$ if it is ``applied'' to $t=\langle n,n,n\rangle$ itself. Starting with the same tensor $t$, however, the Universal  method can in principle give much better bounds on $\omega$ than the Solar or Galactic methods applied to the same $t$.

For a tensor $T$, let $\omega_g(T)$ be the best bound on $\omega$ that one can obtain by applying the Galactic method to $T$. We define a class of {\em generalized $CW_q$} tensors that contain $CW_q$ and many more tensors related to it, such as the rotated tensor used in \cite{almanitcs}. Our {\bf main result} is:
\begin{theorem}[Informal]
There is a universal constant $\ell>2$ independent of $q$ so that for every one of the generalized $CW_q$ tensors $T$, $\omega_g(T)\geq \ell$.
\end{theorem}
{\bf Thus, if one uses a generalized CW tensor, even in the limit and even if one uses the Galactic method subsuming all known approaches, one cannot prove $\omega=2$.}

To prove this result, we develop several tools for proving lower bounds on $\omega_g(T)$ for structured tensors. Most are relatively simple combinatorial arguments but are still powerful enough to show strong lower bounds on $\omega_g(T)$.

We also study the relationship between the generalized $CW$ tensors and the structure tensors of group algebras. We show several new results:
\begin{enumerate}
\item {\bf A Limit on the Group-Theoretic Approach.} The original $CW_q$ tensor is not a sub-tensor (and hence also not a monomial degeneration) of the structure tensor $T_G$ of $\C[G]$ for any $G$ of order $<2q$ when (a) $G$ is abelian and $q$ arbitrary, or (b) $G$ is non-abelian and $q \in \{3,4,5,6,7,8,9\}$.
Note that $CW_q$ for these small values of $q$ are of particular interest: the best known  bounds on $\omega$ have been proved using $q<7$. This shows that lower bound techniques based on tri-colored sum-free sets and group tensors cannot be easily applied to $CW_q$.

\item {\bf All Finite Groups Suffice for Current $\omega$ Bounds.} Every finite group $G$ has a monomial degeneration to some generalized CW tensor of parameter $q=|G|-2$. Thus, applying the Galactic method on $T_G$ for {\em every} $G$ (with sufficiently small asymptotic rank, i.e. $\tilde{R}(T_G)=|G|$) can yield the current best bounds on $\omega$.

\item {\bf New Tri-Colored Sum-Free Set Constructions.} For every finite group $G$, there is a constant $c_{|G|} > 2/3$ depending only on $|G|$ such that its $n$th tensor power $G^n$ has a tri-colored sum-free set of size at least $|G|^{c_{|G|}n - o(n)}$. For moderate $|G|$, the constant $c_{|G|}$ is quite a bit larger than $2/3$. To our knowledge, such a general result was not known until now.
\end{enumerate}
For more details on our results, see Section~\ref{sec:overview} below.

\section{Overview of Results and Proofs}
\label{sec:overview}
In this section, we give an outline of our techniques which are used to prove our main result: that there exists a universal constant $c>2$ such that the Galactic method, when applied to any generalized Coppersmith-Winograd tensor, cannot prove a better upper bound on $\omega$ than $c$. We will assume familiarity with standard notions and notation about tensors related to matrix multiplication algorithms in this section; we refer the reader to the Preliminaries, in Section~\ref{sec:prelims}, where these are defined. For a tensor $T$, we will write $\omega_g(T)$ to denote the best upper bound on $\omega$ which can be achieved using the Galactic method applied to $T$.

\paragraph{Step 1: The Relationship Between Matrix Multiplication and Independent Tensors.}

In Section~\ref{sec:independent}, we begin by laying out the main framework for proving lower bounds on $\omega_g(T)$. 
The key is to consider a different property of $T$, the \emph{asymptotic independence number of $T$}, denoted $\tilde{I}(T)$. Loosely, $\tilde{I}(T)$ gives a measure of how large of an independent tensor $T^{\otimes n}$ can monomial degenerate into for large $n$. From the definition, we will get a simple upper bound $\tilde{I}(T) \leq \tilde{R}(T)$, the \emph{asymptotic rank} of $T$.
By constructing upper bounds on $\tilde{I}(T)$, we will show in Corollary~\ref{cor:omegaandi} that:
\begin{itemize}
    \item For any tensor $T$, if $\omega_g(T)=2$, then $\tilde{I}(T) = \tilde{R}(T)$, and moreover,
    \item For every constant $s<1$, there is a constant $w>2$ (which is increasing as $s$ decreases), such that if $\tilde{I}(T) < \tilde{R}(T)^s$, then $\omega_g(T) \geq w$.
\end{itemize}
Hence, upper bounds on $\tilde{I}(T)$ give lower bounds on $\omega_g(T)$. We will thus present a number of different ways to prove upper bounds on $\tilde{I}(T)$ in the next steps.

\paragraph{Step 2: Partitioning Tools for Upper Bounding $\tilde{I}$.}

In Section~\ref{sec:partitioning}, we present our first suite of tools for proving upper bounds on $\tilde{I}(T)$. These tools are based on finding `local' combinatorial properties of the tensor $T$ which imply that $\tilde{I}(T)$ can't be too large. They are loosely summarized as follows; in the below, let $T$ be a tensor over $X,Y,Z$:

\begin{itemize}
    \item Theorem~\ref{thm:removeanx}: Let $S$ be any subset of the $X$-variables of $T$, and let $A$ be the tensor $T$ restricted to $S$ (i.e. $T$ with all the variables in $X \setminus S$ zeroed out). If $\tilde{I}(A)$ is sufficiently smaller than $|S|$, then $\tilde{I}(T) < |X| \leq \tilde{R}(T)$. 
    
    In other words, if $A$ has a sufficiently small $\tilde{I}(A)$ so that it is relatively far away from being able to prove $\omega_g(A)=2$, then no matter how we complete $A$ to get to $T$, the tensor $T$ will still not be able to prove $\omega_g(T)=2$.
    
    \item Theorem~\ref{thm:probs}: If $T$ is a tensor such that $\tilde{I}(T)$ is close to $\tilde{R}(T)$, then there is a probability distribution on the terms of $T$ such that each $X$, $Y$, and $Z$ variable is assigned almost the same probability mass.
    
    For many tensors of interest, one or more of the variables `behave differently' from the rest, and this can be used to prove that such a probability distribution cannot exist. For one example, we prove in Corollary~\ref{cor:corners} that if $T$ is a tensor with two `corner terms' -- terms $x_q y_1 z_1, x_1 y_q z_1 \in T$ such that no other term in $T$ contains either $x_q$ or $y_q$ -- then, $\tilde{I}(T)<\tilde{R}(T)$.
    
    These `corner terms' are actually quite common in tensors which have been analyzed with the Laser Method. For instance, one of the main improvements of Coppersmith-Winograd~\cite{coppersmith} over Strassen~\cite{laser} was noticing that the border rank expression of Strassen could be augmented by adding in three  corner terms, resulting in the Coppersmith-Winograd tensor.
    
    \item Theorem~\ref{thm:measures}: For a tensor $T$ over variables $X,Y,Z$, where each of these variables appears in the support of $T$, we define the \emph{measure} of $T$, denoted $\mu(T)$, by $\mu(T) := |X| \cdot |Y| \cdot |Z|$. Suppose the terms of $T$ can be partitioned\footnote{We mean `partitioned' as in a set partition, not any restricted notion like the `block partitions' of the Laser Method.} into tensors $T_1, \ldots, T_k$. Then, $\tilde{I}(T) \leq (\mu(T_1))^{2/3} + \cdots + (\mu(T_k))^{2/3}$.
    
    This gives a generalization of the basic inequality that $\tilde{I}(T) \leq \min\{ |X|, |Y|, |Z|\}$. Whenever $T$ can be partitioned up into parts which each do not have many of one or more type of variable, we can get a nontrivial upper bound on $\tilde{I}(T)$. Many natural border rank expressions naturally give rise to such partitions, as do the `blockings' used in the Laser method.
\end{itemize}

As we will see, $\tilde{I}$ is neither additive nor multiplicative, i.e. there are tensors $A$ and $B$ such that $\tilde{I}(A+B) \gg \tilde{I}(A) + \tilde{I}(B)$, and tensors $C$ and $D$ such that $\tilde{I}(C \otimes D) \gg \tilde{I}(C) \cdot \tilde{I}(D)$. One of the main components of the proofs of correctness of each of the three tools above will be narrowing in on classes of tensors $A$ and $B$ such that $\tilde{I}(A+B)$ is not too much greater than $\tilde{I}(A)+\tilde{I}(B)$, or classes of tensors $C$ and $D$ such that $\tilde{I}(C \otimes D)$ is not too much greater than $\tilde{I}(C) \cdot \tilde{I}(D)$. Our proofs will then manipulate our tensors using partitionings so that they fall into these classes.

\paragraph{The Main Result.}
The three partitioning tools are designed to be useful for proving nontrivial upper bounds on $\tilde{I}$ for general classes of tensors. They are especially well-suited to tensors which have structures that make them amenable to known techniques like the Laser Method. In particular, we will ultimately show that any generalized Coppersmith-Winograd tensor has \emph{all three} of these properties. Indeed, our main result, Theorem~\ref{thm:main}, follows from these tools: For any generalized CW tensor $T$, a lower bound on $\omega_g(T)$ for small $q$ will follow from Corollary~\ref{cor:corners}, and a lower bound on $\omega_g(T)$ as $q$ gets large (but such that the bound gets larger as $q$ increases, not smaller) will follow from either Theorem~\ref{thm:removeanx} or Theorem~\ref{thm:measures}.

\paragraph{Bounds on $\tilde{I}$ for Group Tensors.} 
In addition to the above, we also study group tensors. For a finite group $G$, we call the structural tensor $T_{\C[G]}$ of the group algebra $\C[G]$ the \emph{group tensor $T_G$ of $G$}. We are able to achieve both nontrivial upper bounds and lower bounds on $\tilde{I}(T_G)$ for any finite group $G$, including non-abelian groups.

\paragraph{Upper Bounds on $\tilde{I}(T_G)$.}
We first show that for any finite group $G$, we have $\tilde{I}(T_G) < |G| \leq \tilde{R}(T_G)$, and hence $\omega_g(T_G)>2$. In other words, no fixed group $G$ can yield $\omega=2$ by using the Galactic method applied to $T_G$. By comparison, the Group Theoretic approach for $G$ can be viewed as analyzing $T_G$ using a particular technique within the Solar method (see Section \ref{sec:gtm} for more details). This therefore generalizes a remark which is already known within the Group Theoretic community \cite{blasiak2017groups}: that the Group Theoretic approach (using the so-called `Simultaneous Triple Product Property') cannot yield $\omega=2$ using any fixed finite group $G$. It does not, however, rule out using a sequence of groups whose lower bounds approach $2$.

Our proof begins by proving a generalization of a remark from \cite{almanitcs}: that lower bounds on $\tilde{I}(T_G)$ give rise to constructions of `tri-colored sum-free sets' in $G^n$ for sufficiently large integer $n$ (\cite{almanitcs} proved this when $G$ is a cyclic group, although our proof is almost identical). Tri-colored sum-free sets are objects from extremal combinatorics which have been studied extensively recently. We will, in particular, use a recent result of Sawin \cite{sawin2017bounds}, who showed that for any finite group $G$, there is a sufficiently large $n$ such that $G^n$ does not have particularly large tri-colored sum-free sets. 

We give this proof in Section~\ref{sec:grouptools}. In that section, we also show that there are natural tensors, like the Coppersmith-Winograd tensors used to give the best known upper bounds on $\omega$, which \emph{cannot} even be written as sub-tensors of relatively small group tensors. In other words, the high-powered hammer that $\omega_g(T_G)>2$ cannot be used to give lower bound for every tensor of interest, and other techniques like the combinatorial partitioning techniques from step 2 above are needed.

\paragraph{Lower Bounds on $\tilde{I}(T_G)$}
Although our main framework involves proving upper bounds on $\tilde{I}(T)$ for tensors $T$ in order to prove lower bounds on $\omega_g(T)$, step 1 of our proof actually involves constructing lower bounds on $\tilde{I}(T)$ when $T$ has a monomial degeneration to a matrix multiplication tensor. In Section~\ref{sec:tcsfslb}, we use this to give lower bounds on $\tilde{I}(T_G)$ for any finite group $G$. 

We show in Theorem~\ref{thm:anygroupcw} that for \emph{any} finite group $G$, there is a monomial degeneration of $T_G$ into a \emph{generalized} Coppersmith-Winograd tensor of parameter $|G|-2$. We will see that the Laser method applies just as well to any generalized Coppersmith-Winograd tensor of parameter $|G|-2$ as it does to the original $CW_{|G|-2}$, and so the best-known approach for finding matrix multiplication tensors as monomial degenerations of a tensor can be applied to \emph{any} group tensor $T_G$ as well. Two important consequences of this are:
\begin{enumerate}
    \item For any group $G$ such that $\tilde{R}(T_G) = |G|$, we can use the Galactic method to achieve the best known upper bound on $\omega$ (that is known from $CW_{|G|-2}$) by using $T_G$ as the underlying tensor instead of the Coppersmith-Winograd tensor. We think this has exciting prospects for designing new matrix multiplication algorithms; see Remark~\ref{rem:TGomegaUB} for further discussion.
    \item Once $T_G$ has been monomial degenerated into a Coppersmith-Winograd tensor, and thus a matrix multiplication tensor, we can then apply the tools from step 1 above to show that $T_G$ has a monomial degeneration to a relatively large independent tensor. In particular, we show that for any group $G$, $\tilde{I}(T_G) \geq |G|^{c_{|G|}}$ for some constant $c_{|G|}>2/3$ which depends only on $|G|$. Combining this with the connection between $\tilde{I}(T_G)$ and tri-colored sum-free sets in $G$, we see that for any finite group $G$, $G^n$ has a tri-colored sum-free set of size at least $|G|^{c_{|G|} n - o(n)}$. See Theorem~\ref{thm:CWIlb} and the remainder of Section~\ref{sec:tcsfslb} for the details. We will find that $c_{|G|}$ is much bigger than $2/3$ for reasonable $|G|$; for instance, that $c_{|G|} > 3/4$ for $|G| < 250$.
\end{enumerate}

\section{Preliminaries} \label{sec:prelims}

\subsection{Tensor Notation and Definitions}

Let $X = \{x_1, \ldots, x_q\}$, $Y = \{y_1, \ldots, y_r\}$, and $Z = \{z_1, \ldots, z_s \}$ be three sets of formal variables. A \emph{tensor over $X,Y,Z$} is a trilinear form
$$T = \sum_{x_i \in X, y_j \in Y, z_k \in Z} T_{ijk} x_i y_j z_k,$$
where the $T_{ijk}$ coefficients come from an underlying field $\F$. 
One writes $T\in \F^q\otimes \F^r\otimes \F^s$, and the triads $x_i y_j z_k$ are typically written $x_i\otimes y_j\otimes z_k$; we omit the $\otimes$ for ease of notation. 
When $X,Y$, and $Z$ are clear from context, we will just call $T$ a \emph{tensor}. 
The \emph{support} of a tensor $T$ are all triples $(i,j,k)$ for which $T_{ijk}\neq 0$.
The \emph{size} of a tensor $T$, denoted $|T|$, is the size of its support.
We will write $x_i y_j z_k \in T$ to denote that $(i,j,k)$ is in the support of $T$, and in this case we call $x_i y_j z_k$ a \emph{term} of $T$. We will call elements of $X$ the `$x$-variables of $T$', and similarly for $Y$ and $Z$.

If $A\in \F^k\otimes \F^m\otimes \F^n$ and $B\in \F^{k'}\otimes \F^{m'}\otimes \F^{n'}$, then the \emph{tensor product (or Kronecker product) of $A$ and $B$}, denoted $A \otimes B$, is a tensor in $\F^{k\times k'}\otimes \F^{m\times m'}\otimes \F^{n\times n'}$ over new variables $\bar{X},\bar{Y},\bar{Z}$ 
given by
$$A \otimes B = \sum_{\substack{(i,i') \in [k] \times [k'] \\ (j, {j'}) \in [m] \times [m'] \\ (k, {k'}) \in  [n] \times [n']}} A_{ijk} B_{i'j'k'} \bar{x}_{ii'} \bar{y}_{jj'} \bar{z}_{kk'}.$$

The $n$th tensor power of a tensor $A$, denoted $A^{\otimes n}$, is the result of tensoring $n$ copies of $A$ together, so $A^{\otimes 1}= A$, and $A^{\otimes n} = A \otimes A^{\otimes (n-1)}$.

Intuitively, if $A$ is over $X,Y,Z$ and $B$ is over $X',Y',Z'$, then the variables $\bar{x}_{ii'}, \bar{y}_{jj'},\bar{z}_{kk'}$ of $A\otimes B$ can be viewed as pairs of the original variables $(x_i, x'_{i'}) (y_j, y'_{j'}) (z_k, z'_{k'})$. We will use this view in some of our proofs. For instance, when considering $A^{\otimes n}$ we will often view the $x$,$y$ and $z$ variables of $A^{\otimes n}$  as ordered $n$-tuples of $x$,$y$ and $z$ variables of $A$. Then we can discuss for instance, in how many positions of an $x$ variable of $A^{\otimes n}$, the variable $x_i$ of $A$ appears.

\subsubsection{Tensor Rank}

A tensor $T$ has \emph{rank one} if there are values $a_i \in \F$ for each $x_i \in X$, $b_j \in \F$ for each $y_j \in Y$, and $c_k \in \F$ for each $z_k \in Z$, such that $T_{ijk} = a_i b_j c_k$, or in other words,
$$T=\sum_{x_i \in X, y_j \in Y, z_k \in Z} a_i b_j c_k \cdot x_i y_j z_k = \left( \sum_{x_i \in X} a_i x_i \right) \left( \sum_{y_j \in Y} b_j y_j \right) \left( \sum_{z_k \in Z} c_k z_k \right).$$
More generally, the \emph{rank of $T$}, denoted $R(T)$, is the smallest nonnegative integer $m$ such that $T$ can be written as the sum of $m$ rank-one tensors.

Let $\lambda$ be a formal variable, and suppose $T$ is a tensor over $X,Y,Z$. 
The \emph{border rank of $T$}, denoted by\footnote{Much of the literature uses $\underline{R}$ for border rank; we instead use $\bar{R}$ for ease of notation.} $\bar{R}(T)$, is the smallest $r$ such that there is a tensor $\mathcal{T}$ with coefficients $\mathcal{T}_{ijk}$ in $\F[\lambda]$ (polynomials in $\lambda$), so that for every setting of $\lambda\in \F$, $\mathcal{T}$ evaluated at $\lambda$ has rank $r$, and so that 
there is an integer $h\geq 0$ for which:
\[\lambda^h T = \mathcal{T} + O(\lambda^{h+1}).\]
The above notation means that for every $i,j,k$, the polynomial $\mathcal{T}_{ijk}$ over $\lambda$ has no monomials with $\lambda^j$ with $j<h$, and the coefficient in front of $\lambda^h$ in $\mathcal{T}_{ijk}$ is exactly $T_{ijk}$. In a sense, the family of rank $r$ tensors $\mathcal{T}\lambda^{-h}$ for $\lambda\neq 0$ can get arbitrarily close to $T$ -- if $\F=\R$, then we could think of taking $\lambda\rightarrow 0$ and then $\mathcal{T}\lambda^{-h}\rightarrow T$.

The {\em asymptotic rank} of a tensor $T$ is defined as $\tilde{R}(T):=\lim_{n\rightarrow \infty} (R(T^{\otimes n}))^{1/n}$. The limit exists and equals $\inf_{n\rightarrow \infty} (R(T^{\otimes n}))^{1/n}$.
It is known that for any tensor $T$, $$R(T) \geq \bar{R}(T) \geq \tilde{R}(T),$$
and that each of these inequalities can be strict\footnote{For example, the first inequality is strict for the Coppersmith-Winograd tensor, and the second inequality is strict for the $2 \times 2 \times 2$ matrix multiplication tensor. Both of these tensors will be defined shortly.}. One of the most common ways to show asymptotic rank upper bounds is to give border rank upper bounds, frequently using a tool called a `monomial degeneration' which we will define shortly.

The tensor $\langle r\rangle$ in $\F^r\otimes \F^r\otimes \F^r$ is defined as follows: for all $i\in \{1,\ldots,r\}$, $\langle r\rangle_{i,i,i}=1$ and for all other entries $\langle r\rangle_{i,j,k}=0$. $\langle r\rangle$ clearly has rank $r$; it is the natural generalization of an identity matrix. If a tensor $T$ is equivalent to $\langle r\rangle$ up to permutation of the indices, we say that $T$ is an independent tensor of size $|T|=r$.

\subsubsection{Sub-Tensors and Degenerations}

We call a tensor $t$ a {\em sub-tensor} of a tensor $t'$, denoted by $t\subseteq t'$, if $t$ can be obtained from $t'$ by removing triples from its support, i.e. for every $i,j,k$, either $t_{i,j,k}=t'_{i,j,k}$, or $t_{i,j,k}=0$.

A tensor $t\in \F^k\otimes \F^m\otimes \F^n$ is a {\em restriction} of a tensor $t'\in \F^{k'}\otimes \F^{m'}\otimes \F^{n'}$, written $t\leq t'$, if there are homomorphisms $\alpha:\F^k\mapsto \F^{k'}$, $\beta:\F^m\mapsto \F^{m'}$, and $\gamma:\F^n\mapsto \F^{n'}$, so that $t=(\alpha\otimes\beta\otimes\gamma)t'$.\footnote{The notation $(\alpha\otimes\beta\otimes\gamma)t$ means the following. Let $t=\sum_{\ell=1}^r (\sum_i a^\ell_i x_i)(\sum_j b^\ell_j y_j) (\sum_k c^\ell_k z_k) = \sum_{\ell=1}^r (a^\ell\cdot x)(b^\ell\cdot y)(c^\ell\cdot z)$ be any decomposition of $t$ into a sum of rank $1$ tensors, where $a^\ell=(a^\ell_1,\ldots,a^\ell_k)\in \F^k,b^\ell=(b^\ell_1,\ldots,b^\ell_m)\in\F^m,c^\ell=(c^\ell_1,\ldots,c^\ell_n)\in \F^n$. Then $(\alpha\otimes\beta\otimes\gamma)t := \sum_{\ell=1}^r (\alpha(a^\ell)\cdot x)(\beta(b^\ell)\cdot y) (\gamma(c^\ell)\cdot z)$ is well-defined.}
The rank of $t$ is $\leq r$ if and only if $t\leq \langle r\rangle$.

A special type of restriction is the so called {\em zeroing out} (also called {\em combinatorial restriction}): let $t$ be a tensor over $X,Y,Z$; $t'$ is a zeroing out of $t$ if it is obtained by selecting $X'\subseteq X, Y'\subseteq Y, Z'\subseteq Z$ and setting to zero all $x_i\in X\setminus X', y_j\in Y\setminus Y', z_k\in Z\setminus Z'$; thus, $t'$ is a tensor over $X',Y',Z'$ and it equals $t$ on all triples over these sets.

A {\em degeneration} $t'\in \F^{k'}\otimes \F^{m'}\otimes \F^{n'}$ of a tensor $t\in \F^k\otimes \F^m\otimes \F^n$, written $t'\trianglelefteq t$, is obtained as follows. Similarly to the definition of border rank, let $\lambda$ be a formal variable.
We say that  $t'\trianglelefteq t$ if there exist
$q\in\N$, $A(\lambda)\in\F^{k'\times k}, B(\lambda)\in\F^{m'\times m}, C(\lambda)\in\F^{n'\times n}$ matrices with entries which are polynomials in $\lambda$ (i.e. in $\F[\lambda]$), so that 
\[\lambda^q t' = (A(\lambda)\otimes B(\lambda)\otimes C(\lambda)) t + O(\lambda^{q+1}).\]
Similarly to the relationship between rank and restriction, the border rank of $t$ is at most $r$ if and only if $t\trianglelefteq \langle r\rangle$.

A special type of degeneration is the so called {\em monomial degeneration} (also called {\em combinatorial degeneration} or {\em toric degeneration}), in which the matrices $A(\lambda),B(\lambda),C(\lambda)$ have entries that are monomials in $\lambda$. An equivalent definition of monomial degeneration \cite{almanitcs} is as follows: suppose that $t'$ is a tensor over $\F^k\otimes \F^m\otimes \F^n$, $t\subseteq t'$ is a sub-tensor,  and there are functions $a : [k] \to \Z$, $b : [m] \to \Z$, and $c : [n] \to \Z$ such that (1) whenever $t'_{ijk} \neq 0$,  $a(i) + b(j) + c(k) \geq 0$, (2) if $a(i) + b(j) + c(k) = 0$, then $t_{i,j,k}=t'_{i,j,k}$, and (3) if $t_{ijk} \neq 0$, then $a(i) + b(j) + c(k) = 0$.

\subsubsection{Structural Properties of Tensors}

We say that a tensor $T$ is {\em partitioned} into tensors $T_1,\ldots, T_\ell$, if $T=T^1+\ldots+T^\ell$, and for every triple $i,j,k$, there is a $w$ such that $T^w_{i,j,k}= T_{i,j,k}$ and for all $w'\neq w$,  $T^w_{i,j,k}=0$. In other words, the triples in the support of $T$ are partitioned into $\ell$ parts, forming $\ell$ tensors summing to $T$.\footnote{Note that this notion of partitioning is more general than `block partitioning' from the Laser Method (which we define shortly), although `block partitioning' is occasionally referred to as just `partitioning' in the literature.}

A {\em direct sum} of two tensors $t$ and $t'$ over disjoint variable sets $X,Y,Z$ and $X',Y',Z'$, $t\oplus t'$ is the tensor on variable sets $X\cup X',Y\cup Y',Z\cup Z'$ which is exactly $t$ on triples in $X\times Y\times Z$, exactly $t'$ on triples in $X'\times Y'\times Z'$, and is $0$ on all other triples. In contrast, a regular sum $t+t'$ could have $t$ and $t'$ share variables.

Similar to how a $k \times m$ matrix in $\F^k \otimes \F^m$ can be viewed as a linear map from $\F^k$ to $\F^m$, a tensor $t=\sum_{i,j,k} t_{i,j,k} x_iy_jz_k$ in $\F^k\otimes \F^m\otimes \F^n$ can be viewed as a linear map $T_X : \F^k \to  \F^m \otimes \F^n$ which maps $x_i$ to $\sum_{y_j \in Y, z_k \in Z} t_{ijk} y_j z_k\in \F^m\otimes \F^n$.
One can also exchange the roles of the $x,y$ and $z$ variables, so that $t$ can also be viewed as a linear map $T_Y : \F^m \to  \F^k \otimes \F^n$, or a linear map $T_Z: \F^n \to  \F^k \otimes \F^m$. The tensor $t$ is called {\em concise} if $T_X,T_Y,T_Z$ are injective.
It is not hard to see that $\bar{R}(t)\geq \max\{Rank(T_X),Rank(T_Y),Rank(T_Z)\}$, so that for concise tensors, $\bar{R}(t)\geq \max\{|X|,|Y|,|Z|\}$. All the explicit tensors we will discuss throughout this paper, including the tensor of matrix multiplication, and the Coppersmith-Winograd tensor, are concise.

\subsection{The Matrix Multiplication Tensor and Methods for Analyzing $\omega$}
Let $m,n,p\geq 1$ be integers. The tensor of $m\times n$ by $n\times p$ matrix multiplication over a field $\F$, denoted by $\langle m,n,p\rangle$, lies in $\F^{m\times n}\otimes \F^{n\times p}\otimes \F^{p\times m}$, and in trilinear notation looks like this:
\[\langle m,n,p\rangle = \sum_{i=1}^m\sum_{j=1}^n\sum_{k=1}^p x_{ij}y_{jk}z_{ki}.\]

The theory of matrix multiplication algorithms is concerned with determining the value $\omega$, defined as $\omega:=\inf \{c\in \R~|~R(\langle n,n,n\rangle)\leq O(n^c)\}$. (As shown by Coppersmith and Winograd~\cite{CoppersmithW82}, $\omega$ is a limit point that cannot be achieved by any single algorithm.)

Getting a handle on $\omega$ has been difficult. Over the years various methods have been developed to obtain better understanding of the rank of $\langle n,n,n\rangle$. The basic idea of all methods is as follows: Although we do not know what the true rank of $\langle n,n,n\rangle$ is, as $n$ grows, there are many other tensors for which we know their rank and even their asymptotic rank exactly. Hence, the approach is, take a tensor $t$ whose asymptotic rank $\tilde{R}(t)$ we understand, take a large tensor power $t^{\otimes N}$ of $t$, and ``embed'' $\langle f(N),f(N),f(N)\rangle$ into $t^{\otimes N}$ so that the embedding shows that  $\tilde{R}(\langle f(N),f(N),f(N)\rangle)\leq \tilde{R}(t)^N$. From this inequality we can get a bound on $\omega$, by taking $N$ to $\infty$. More generally, by Sch\"onhage's Asymptotic Sum Inequality (Theorem~\ref{thm:tauthm} below), it is actually sufficient to embed the direct sum of many smaller copies of matrix multiplication tensors into $t^{\otimes N}$ to get a similar bound on $\omega$.

The way in which the approaches differ is mainly in how the embedding into $t^{\otimes N}$ is obtained. All known approaches to embed a matrix multiplication tensor into a tensor power $t^{\otimes N}$ of some other tensor $t$ actually all zero out variables in $t^{\otimes N}$ and argue that after the zeroing out, the remaining tensor is a matrix multiplication tensor. 

There are two main approaches for obtaining good bounds on $\omega$ via zeroing out $t^{\otimes N}$: the laser method and the group theoretic approach. We will describe them both shortly.

Zeroing out is a very restricted border-rank preserving operation on a tensor. The most general embedding of a matrix multiplication tensor into $t^{\otimes N}$ would be a potentially complicated degeneration of $t^{\otimes N}$. In fact, in this case, since every border rank $q$ tensor is a degeneration\footnote{This folklore fact follows from inverting the DFT over cyclic groups; see eg. \cite[Section~3.1]{almanitcs}.} of the structure tensor for addition modulo $q$, $T_q=\sum_{i=0}^{q-1}\sum_{j=0}^{q-1} x_i y_j z_{i+j\bmod q}$, it would suffice to find a degeneration of $T_q^{\otimes n}$ into a large matrix multiplication tensor, for large $n$. Unfortunately, we currently do not have techniques to find good degenerations. We call this hypothetical method the {\bf Universal} method.

Instead of considering arbitrary degenerations of $t^{\otimes n}$, we could instead consider monomial degenerations of $t^{\otimes n}$ into a large matrix multiplication tensor. This approach would subsume both the Laser Method and the Group Theoretic approach.
Although again there are no known techniques to obtain better monomial degenerations than zeroing outs, monomial degenerations seem easier to argue about than arbitrary degenerations. We call the method of finding the optimal (with respect to bounding $\omega$) monomial degeneration of a tensor power into a matrix multiplication tensor, the {\bf Galactic} method.
(Reaching the end of our Galaxy is more feasible than seeing the entire Universe.) 
To complete the analogy, we can call the method using zeroing outs the {\bf Solar} method (i.e. exploring the Solar System).

The Solar method subsumes the Group Theoretic Approach and the Laser Method, but is more general, and current techniques do not suffice to find the optimal zeroing-out of $t^{\otimes n}$ into matrix multiplication even for simple tensors.  Our lower bounds will be not only for the Solar method, but also for the Galactic method which is even more out of reach for the current matrix multiplication techniques.

To be clear, the Solar method, Galactic method, and Universal method, give us successively more power when analyzing specific tensors. For example, it may be the case that for a specific tensor $T$, the Solar method applied to $T$ cannot get as low an upper bound on $\omega$ as the Universal method applied to $T$ can. This captures the known methods to get bounds on $\omega$ by using tensors like the Coppersmith-Winograd tensor or a group tensor, which we will define shortly. The three different methods will trivially give the same bound, $\omega$, when applied to matrix multiplication tensors themselves, but this is not particularly interesting: the entire point of these different methods is that the asymptotic rank of matrix multiplication tensors is not well-understood, and applying the methods to other tensors can help us get better bounds on it.

We will now describe the two approaches that follow the Solar method.

\subsection{The Laser Method}
Strassen~\cite{laser} proposed a method for embedding a matrix multiplication tensor into a large tensor power of a starting tensor. He called it the {\em Laser Method}. In this method, we start with a tensor $t$ over variables $X$, $Y$, $Z$ of asymptotic rank $q$, where say $|X|=q$, so that $t$ has essentially optimal asymptotic rank. The variable sets are then partitioned into {\em blocks}: $X=X_1\cup\ldots\cup X_a$, $Y=Y_1\cup\ldots\cup Y_b$, $Z=Z_1\cup\ldots,Z_c$. Define by $t_{IJK}$ the sub-tensor of $t$ obtained by zeroing-out all variables $x\notin X^I$, $y\in Y^J$, $z\in Z^K$. We obtain a partitioning
\[t=\sum_{I\in [a],J\in [b],K\in [c]} t_{IJK}.\]

Ideally, the {\em constituent} tensors $t_{IJK}$ should be matrix multiplication tensors, but this is not necessary. 

In the large tensor power $t^{\otimes N}$, one then is allowed to zero out variables $\bar{x}_i$, $\bar{y}_j$ and $\bar{z}_k$ (removing all triples containing them). This zeroing out is not arbitrary, however: if some variable, say $\bar{x}_i$ is zeroed out, consider its index $i$ -- it is a sequence of length $N$ of original indices $i[1],i[2],\ldots,i[N]$. Say that $x_{i[z]}\in X_{I(z)}$ (i.e. $I(z)$ is the block that $\bar{x}_i$ uses in its $z$th coordinate). Then every other $x$ variable, $\bar{x}_{i'}$ for which $x_{i'[z]}\in X_{I(z)}$ for all $z$, must be zeroed out as well. That is, variables with the same block sequence must either all be kept or all zeroed out.

One considers such possible zeroing outs and attempts to argue that one of them leaves exactly a direct sum of matrix multiplication tensors (possibly of different dimensions). Then one uses the {\em asymptotic sum inequality} of Sch\"onhage \cite{Sch81} to obtain a bound on $\omega$:

\begin{theorem}[Asymptotic Sum Inequality \cite{Sch81}]
If $\bigoplus_{i=1}^p \langle k_i,m_i,n_i\rangle$ has border rank $\leq r$, and $r>p$, then $\omega\leq 3\tau$, where $\sum_{i=1}^p (k_im_in_i)^\tau = r$. \label{thm:tauthm}
\end{theorem}

Looking at Sch\"onhage's proof of the asymptotic sum inequality, however, we see that what it is actually doing is, taking a large tensor power of $\bigoplus_{i=1}^p \langle k_i,m_i,n_i\rangle$ and {\em zeroing out} variables to obtain independent copies of the same single matrix multiplication tensor, i.e. $F\odot \langle K,M,L\rangle$. Thus, we can think of the laser method as zeroing out $t^{\otimes N}$ in a block-preserving fashion, to obtain a copies of the same matrix multiplication tensor.

We now turn to the most successful implementation of the Laser Method: the {\bf Coppersmith-Winograd} approach.

The Coppersmith-Winograd (CW) family of tensors is as follows: Let $q\geq 1$ be an integer.
$$CW_q = x_0 y_0 z_{q+1} + x_{q+1} y_0 z_{0} + x_0 y_{q+1} z_{0} + \sum_{i=1}^q (x_i y_0 z_i + x_0 y_i z_i + x_i y_i z_{0}).$$
$CW_q$ is a concise tensor over $\F^{q+2} \otimes \F^{q+2}\otimes \F^{q+2}$, of border rank (and hence also asymptotic rank) $q+2$.

Coppersmith and Winograd~\cite{coppersmith} followed the laser method. The tensors $CW_q$ have a natural partitioning $CW_q = T_{002}+T_{020}+T_{200}+T_{011}+T_{101}+T_{110}$, where $T_{002}=x_0 y_0 z_{q+1}, T_{200}= x_{q+1} y_0 z_{0}, T_{020}= x_0 y_{q+1} z_{0}, T_{101}= \sum_{i=1}^q x_i y_0 z_i, T_{011}=\sum_{i=1}^q x_0 y_i z_i, T_{110}=\sum_{i=1}^q  x_i y_i z_{0}$. 

The partitioning is actually a block partitioning: The $T_{IJK}$ are obtained by blocking the $X$, $Y$ and $Z$ variables into three blocks: the indices $\{0,\ldots,q+1\}$ are blocked into block $0$ containing $\{0\}$, block $1$ containing $\{1,\ldots,q\}$ and block $2$ containing $\{q+1\}$, and then, block $I$ of $X$ (resp. $Y$ and $Z$) contains all $x_i$ (resp. $y_i$ and $z_i$) with $i$ in block $I$ of the indices. Then $T_{IJK}$ is the block tensor formed by the triples with $x$ variables in block $I$, $y$ variables in block $J$ and $z$ variables in block $K$.

The sub-tensors $T_{IJK}$ have two useful properties: (1) they are all matrix multiplication tensors, (2) for each $T_{IJK}$ above, $I+J+K=2$. 

The Coppersmith-Winograd implementation of the laser method uses these properties together with sets excluding $3$-term arithmetic progressions (in conjunction with property (2) above) to decide which blocks of variables to zero out in $CW_q^{\otimes n}$. Since the zeroing out proceeds by zeroing out variables that have the same block sequences, and due to property (1) in the end one obtains a sum of matrix multiplication tensors, and due to the use of sets excluding $3$-term arithmetic progressions one can guarantee that in fact this is a direct sum of many large matrix multiplication tensors. Then one can use the asymptotic sum inequality to obtain a bound on $\omega$. To optimize the bound on $\omega$, one selects the best $q$, which ends up being $q=6$. Coppersmith and Winograd then achieve a slightly better bound on $\omega$ by analyzing the square $CW_q^{\otimes 2}$ in a similar way.

The later improvements on the Coppersmith-Winograd bounds by Stothers~\cite{stothers}, Vassilevska W.~\cite{v12} and Le Gall~\cite{legall} instead used the laser method with the CW tools starting from $CW_q^{\otimes 4}, CW_q^{\otimes 8}$ and $\{CW_q^{\otimes 16}$ and $CW_q^{\otimes 32}\}$, respectively. Each new analysis used different, but related, blockings and partitionings, and each ultimately optimized the resulting bound on $\omega$ by picking $q=5$, and hence using $CW_5$ as the base tensor.

The Coppersmith-Winograd analysis works for any blocking of the variables of a tensor $t$ into blocks with integer names so that there exists an integer $b$ such that for every triple $(I,J,K)$ where $I$ is an $x$-block, $J$ is a $y$-block and $K$ is a $z$-block, $I+J+K=b$. For such a blocking, each constituent tensor $T_{IJK}$ should ideally be a matrix multiplication tensor itself. In recent applications of the method, the tensors $T_{IJK}$ need not be matrix multiplications, but then one needs to perform a Coppersmith-Winograd analysis on them to obtain a bound known as their {\em Value} which roughly says how good they are at supporting matrix multiplication.

The Coppersmith-Winograd approach doesn't exploit very much about the block tensors $T_{IJK}$.
In particular, one can replace each $T_{IJK}$ with another tensor $T'_{IJK}$ over the same sets of variables $X_I,Y_J,Z_K$, as long as $T'_{IJK}$ has the same ``value'', and the modified tensor $T'$ has the same border rank as $T$; the bound on $\omega$ the approach would give would be exactly the same!
When $T_{IJK}$ is a matrix multiplication tensor $\langle a,b,c\rangle$, for instance, one can replace it with another matrix multiplication tensor $\langle a',b',c'\rangle$ as long as the new tensor uses the same variables and $a'b'c'=abc$, and as long as the produced full tensor has the same border rank.
For instance, if we take $T_{110}=\sum_{i=1}^q\sum_{j=1}^q x_iy_jz_0$ and replace it with $\sum_{i=1}^q\sum_{j=1}^q x_iy_{q+1-i}z_0$, then we would get the {\em rotated} $CW_q$ tensor studied in \cite{almanitcs}. This tensor still has rank $q+2$ and this gives the same upper bound on $\omega$ using the CW approach.

We can thus define a family of {\em generalized} CW tensors, $\underline{CW}_q$ as follows. 
\begin{definition}
The family $\underline{CW}_q$ of tensors includes, for every permutation $\sigma\in S_q$, the tensor
\[CW^\sigma_q = (x_0y_0z_{q+1}+x_0y_{q+1}z_0+x_{q+1}y_0z_0) + \sum_{i=1}^q (x_i y_{\sigma(i)} z_0 + x_iy_0z_i+x_0y_iz_i).\]
\end{definition}

We remark that the family above contains all tensors obtained from $CW_q$ by replacing $\sum_{i=1}^q (x_i y_{i} z_0 + x_iy_0z_i+x_0y_iz_i)$ with $\sum_{i=1}^q (x_{\tau(i)} y_{\sigma(i)} z_0 + x_{\alpha(i)}y_0z_{\beta(i)}+x_0y_{\gamma(i)}z_{\delta(i)})$ for any choice of $\alpha,\beta,\gamma,\delta,\sigma,\tau\in S_q$.

The constituent tensor $T_{110}$ of $CW^\sigma_q$ is
 $\sum_{i=1}^q x_{i}y_{\sigma(i)}z_0$, which is still a $\langle 1,q,1\rangle$ tensor. Thus, for any such tensor from the family $\underline{CW}_q$, if its border rank is $q+2$, the Coppersmith-Winograd approach would give exactly the same bound on $\omega$, as with $CW_q$.

\subsection{Group-theoretic approach} \label{sec:gtm}
Cohn and Umans~\cite{cohn2003group} pioneered a new group-theoretic approach for matrix multiplication. The idea is as follows. Take a group $G$ and consider its group tensor defined below. (Throughout this paper, we write groups in multiplicative notation.)

\begin{definition}
For any finite group $G$, the \emph{group tensor of $G$}, denoted $T_G$, is a tensor over $X_G, Y_G, Z_G$ where $X_G := \{ x_g \mid g \in G \}$, $Y_G := \{y_g \mid g \in G\}$, and $Z_G := \{ z_g \mid g \in G \}$, given by
$$T_G := \sum_{g,h \in G} x_g y_h z_{gh}.$$
\end{definition}

(Note that the group tensor of $G$ is really the structure tensor of the group algebra $\C[G]$, often written as $T_{\C[G]}$. We use $T_G$ for ease of notation.)

The group-theoretic approach first bounds the asymptotic rank of $T_G$ using representation theory, as follows. Let $d_u$ be the dimension of the $u$th irreducible representation of $G$ (i.e. the $d_u$s are the character degrees). Then $T_G$ can be seen to degenerate from $\bigoplus_{u=1}^\ell \langle d_u\rangle$. In particular, we get that\footnote{It is more straightforward to see that this holds with inequalities (`$\leq$' instead of `$=$') but in fact equality holds because the degeneration of $T_G$ is invertible, and $\omega$ is defined in terms of the asymptotic rank of matrix multiplication tensors.}
\[\tilde{R}(T_G)=\tilde{R}\left(\bigoplus_{u=1}^\ell \langle d_u,d_u,d_u\rangle\right)= \sum_{u=1}^\ell d_u^\omega.\]

Now suppose that we can find any degeneration (e.g. a zeroing out) of $T_G$ into $\bigoplus_{i=1}^s \langle k_i,m_i,n_i\rangle$. Then, by the asymptotic sum inequality we would get that \[\sum_{i=1}^s (k_i m_i n_i)^{\omega/3} \leq \sum_{u=1}^\ell d_u^\omega.\]

Cohn and Umans defined two properties of subsets of $G$ which yield a zeroing out of $T_G$ into matrix multiplication tensors: (1) the triple product property, so that any $G$ that satisfies it admits a zeroing out into a matrix multiplication tensor, and (2) the simultaneous triple product property, so that any $G$ that satisfies it admits a zeroing out into a direct sum of matrix multiplication tensors.

These properties provide the zeroing out, and the group representation provides the rank bound. The approach is extremely clean to define. The goal is then to find a group with known character degrees, satisfying one of the two triple product properties well, so that the matrix multiplication tensors one can get are large.
Typically one works with a family of groups, parameterized by $n$ (as in $Z_n$ or $S_n$), and then one can pick the $n$ that optimizes the bound on $\omega$, or even take $n$ to $\infty$, e.g. when the groups correspond to tensor powers of some tensor.

We refer the reader to \cite[Section 3.5]{landsberg2017geometry} for more exposition on the Group-theoretic approach and its interpretation as finding a zeroing out of group tensors.

\subsection{Independent Tensors}

In this paper, we will be especially interested in zeroing outs and monomial degenerations from tensors $T$ to independent tensors $\langle r \rangle$. We give a few relevant definitions here.

For a tensor $T$ over $X,Y,Z$, its \emph{independence number}, $I(T)$, is the maximum size of an independent tensor which can result from a zeroing out of $T$. We similarly can define the \emph{asymptotic independence number} of $T$ by
$$\tilde{I}(T) := \limsup_{n \in \N} \left[ I(T^{\otimes n}) \right]^{1/n}.$$
Since a zeroing out cannot increase the number of $x$-variables, $y$-variables, or $z$-variables, we get a simple upper bound $I(T) \leq \min\{ |X|, |Y|, |Z| \}$. It similarly follows that $\tilde{I}(T) \leq \min\{ |X|, |Y|, |Z| \}$. Throughout this paper, we will see a number of tensors which achieve equality in this bound, including all matrix multiplication tensors. In Section \ref{sec:independent}, we will prove this and many other properties of $\tilde{I}$.

\subsection{Tri-colored Sum-free Sets}

A number of recent works (eg. \cite{blasiak, blasiak2017groups, almanitcs}) have explored connections between lower bounds on matrix multiplication algorithms, and a notion from extremal combinatorics called a `tri-colored sum-free set'. In this paper, we will expand upon and generalize this connection as one of our tools for proving lower bounds on $\omega_g(T)$ for various tensors $T$.

\begin{definition}
For a group $G$, a \emph{tri-colored sum-free set in $G$} is a set $S \subseteq G^3$ of triples of elements of $G$ such that:
\begin{itemize}
    \item for all $(a,b,c) \in S$, we have $ab=c$, and
    \item for all $(a_1, b_1, c_1), (a_2, b_2, c_2), (a_3, b_3, c_3) \in S$ which are not all the same triple, we have $a_1 b_2 \neq c_3$.
\end{itemize}
In the literature, tri-colored sum-free sets are sometimes also called \emph{multiplicative matchings}.
\end{definition}

In a recent breakthrough, Ellenberg and Gijswijt \cite{eg} used techniques introduced by Croot, Lev, and Pach \cite{croot2017progression} to show that there is a constant $c<3$ such that tri-colored sum-free sets in $\F_3^n$ have size at most $O(c^n)$. Since then, there has been an explosion of work in the area, and this result has been extended by Sawin \cite{sawin2017bounds} to hold for \emph{all} nontrivial groups $G$, even nonabelain groups:

\begin{theorem}[\cite{sawin2017bounds} Theorem 1] \label{thm:multmatch}
Let $G$ be any nontrivial finite group. There is a constant $\delta < 1$ such that for any positive integer $n$, any tri-colored sum-free set in $G^n$ has size at most $(\delta |G|)^n$.
\end{theorem}

There are a number of families of groups $G$ where even stronger upper bounds than this are known; we refer the reader to the introduction of \cite{blasiak2017groups} for an exposition of these bounds. In Section \ref{sec:grouptools}, we will show how Theorem \ref{thm:multmatch} (and also the aforementioned stronger bounds) can be used to give lower bounds on the $\omega$ bound one can achieve using the Galactic method on a wide range of tensors $T$.

\subsection{Comparison with Slice Rank Bounds}

The work on limitations of the group-theoretic approach typically proceeds by giving upper bounds on the so-called `slice rank' of the tensor $T_G$ of a group $G$. It is known \cite{tao1,tao2} that for any tensor $T$, if $T$ has a degeneration to an independent tensor $D$, then $|D| \leq \text{slice-rank}(T)$. Hence, for some tensor $T$, if one can show an upper bound on $\text{slice-rank}(T^{\otimes n})$ for all $n$, this yields an upper bound on $\omega_u(T)$, the value of $\omega$ which can be achieved using the \emph{Universal} method applied to $T$.

For instance, the limitation result of Sawin~\cite{sawin2017bounds}, Theorem~\ref{thm:multmatch} above, is proved by showing that for every fixed group $G$, there is a $\delta<1$ such that $\text{slice-rank}(T_G^{\otimes n}) < \delta^n |G|^n$, which implies using the connection described above that $\omega_u(T_G) > 2$. In particular, this generalizes our Theorem~\ref{thm:groupomeganot2} in which we show that Sawin's result implies that $\omega_g(T_G) > 2$. Again, we note that since $\delta$ depends on $G$, this does not rule out achieving $\omega=2$ by using the Universal method applied to a sequence of groups whose lower bounds on $\omega_u$ approach $2$.

It is worth asking whether similar slice-rank upper bounds can be used to show a lower bound on $\omega_u(CW_q)$ as well. Indeed, $CW_q$ is easily seen to have slice-rank at most $3$. However, slice-rank is not submultiplicative in general, and in fact it is known that $CW_q^{\otimes n}$ can have slice-rank much more than $3^n$. For instance, the fact that $\omega_s(CW_5) \leq 2.373$ implies that $\text{slice-rank}(CW_5^{\otimes n}) \geq 7^{2n/2.373 - o(n)} \geq 5.15^{n - o(n)}$. It is not clear how to upper bound the slice-rank of $CW_q^{\otimes n}$ in general.

We refer to \cite{blasiak, blasiak2017groups} for formal definitions related to slice-rank and matrix multiplication, as we won't need slice-rank in this paper.

\section{Matrix Multiplication and Independent Tensors} \label{sec:independent}

In this section, we will lay out our main framework for proving lower bounds on what values of $\omega$ can be achieved using different tensors $T$ in the Galactic Method. The main idea is that, to prove such a lower bound for tensor $T$, it is sufficient to give an upper bound on $\tilde{I}(T)$.

\begin{definition}
For a tensor $T$, let $\omega_g(T) \geq 2$ denote the best bound on $\omega$ that one can achieve using the Galactic Method with $T$. Hence, for all tensors $T$, we have $\omega \leq \omega_g(T)$.
\end{definition}

\begin{lemma} \label{lem:omegagdef}
Let $T$ be any tensor. For each positive integers $n,a,b,c$, let $F_{T,n,a,b,c}$ be the largest number of disjoint (sharing no variables) copies of $\langle a,b,c \rangle$ which can be found as a monomial degeneration of $T^{\otimes n}$. Then,
$$\omega_g(T) = 3 \cdot \liminf_{n,a,b,c \in \N} \frac{n \log(r) - \log(F_{T,n,a,b,c})}{\log(abc)}.$$
\end{lemma}

\begin{proof}
$\omega_g(T)$ is defined as the $\liminf$, over all $n$ and all ways to monomial degenerate $T^{\otimes n}$ into a disjoint sum of matrix multiplication tensors, of the corresponding bound on $\omega$ which one gets by applying the asymptotic sum inequality, Theorem~\ref{thm:tauthm}. However, as in the proof of Theorem~\ref{thm:tauthm} (see e.g. \cite[Section 7.2]{Sch81} or \cite[Proof of Theorem 7.5]{blaser}), we can restrict our attention without loss of generality to monomial degenerations into a disjoint sum of matrix multiplication tensors of the same dimensions, i.e. monomial degenerations from $T^{\otimes n}$ to $F_{T,n,a,b,c} \odot \langle a, b, c \rangle$ for all choices of $a,b,c,$ and $n$. Then, by Theorem~\ref{thm:tauthm}, if $T^{\otimes n}$ has a monomial degeneration to $F_{T,n,a,b,c} \odot \langle a, b, c \rangle$, this shows that $\tilde{R}(F_{T,n,a,b,c} \odot \langle a, b, c \rangle) \leq \tilde{R}(T^{\otimes n}) = (\tilde{R}(T))^n$, which yields $\omega_g(T) \leq 3 \log((\tilde{R}(T))^n / F_{T,n,a,b,c}) / \log(a b c)$, as desired.
\end{proof}

We use the following monomial degeneration of matrix multiplication tensors which slightly generalizes Strassen's (from \cite[Theorem 4]{laser}). We prove it here for completeness.

\begin{lemma} \label{lem:strassenmonomialdegen}
For any positive integers $a,b,c$, there is a monomial degeneration of $\langle a,b,c \rangle$ into an independent tensor of size $\frac34 \cdot \frac{abc}{\max\{a,b,c\}}$.
\end{lemma}

\begin{proof}
Assume first that $a = 2m+1$, $b=2n+1$, and $c=2p+1$ are all odd, and assume without loss of generality that $c\geq a,b$. Recall that $$\langle a,b,c \rangle = \sum_{i=-m}^m \sum_{j=-n}^n \sum_{k = -p}^p x_{ij} y_{jk} z_{ki}.$$
We define our monomial degeneration via the maps $\alpha : X \to \Z, \beta : Y \to \Z$, and $\gamma : Z \to \Z$ defined as follows:
\begin{itemize}
    \item $\alpha(x_{ij}) = i^2 + 2ij$,
    \item $\beta(y_{jk}) = j^2 + 2jk$, and
    \item $\gamma(z_{ki}) = k^2 + 2ki$,
\end{itemize}
For any term $x_{ij} y_{jk} z_{ki} \in \langle a,b,c \rangle$, we thus have $\alpha(x_{ij}) + \beta(y_{jk}) + \gamma(z_{ki}) = (i+j+k)^2 \geq 0$. We have equality, and thus the term is included in the result $D$ of the monomial degeneration, if and only if $i+j+k=0$. We can see that if $i+j+k=0$, then any two of $i,j,k$ determines the third, meaning any one of the variables $x_{ij}, y_{jk}, z_{ki}$ determines the other two, and so $D$ is indeed an independent tensor. Finally, there is a triple of $(i,j,k)$, $|i| \leq n, |j| \leq m, |k| \leq p$ with $i+j+k=0$ for each pair $(i,j)$, $|i| \leq n, |j| \leq m$ with $|i+j| \leq p$. Since $p \geq n,m$, we can see there are at least $\frac34 ab$ such pairs, as desired. The cases where $a,b,c$ are not all odd are similar.
\end{proof}

Finally we need a Lemma relating monomial degenerations to independent tensors and zeroing-outs to independent tensors, which is a special case of a result of \cite{almanitcs}:

\begin{lemma}[\cite{almanitcs} Lemma 5.1] \label{lem:awmondegtozero}
Suppose $A$ is a tensor which has a monomial degeneration into $f$ independent triples. Then, for positive integers $n$, $A^{\otimes n}$ has a zeroing out into $\Omega(f^n / n^2) = f^{n - o(n)}$ independent triples. 
\end{lemma}

\begin{corollary} \label{cor:mondegtoit}
For any tensor $T$ and positive integer $n$, if $T^{\otimes n}$ has a monomial degeneration to an independent tensor of size $f$, then $\tilde{I}(T) \geq f^{1/n}$.
\end{corollary}

\begin{proof}
By Lemma \ref{lem:awmondegtozero}, for any $\delta>0$, there is a positive integer $m$ such that $(T^{\otimes n})^{\otimes m}$ has a zeroing out into $f^{m(1-\delta)}$ independent triples, which means $I(T^{\otimes nm}) \geq f^{m(1-\delta)}$ and hence $\tilde{I}(T) \geq f^{(1-\delta)/n}$. 
\end{proof}

We similarly get:

\begin{corollary} \label{lem:IAltIB}
For any tensors $A$ and $B$, if $A$ is a monomial degeneration of $B$, then $\tilde{I}(A) \leq \tilde{I}(B)$.
\end{corollary}

Combining our results so far shows that matrix multiplication tensors have large asymptotic independence numbers:

\begin{lemma} \label{lem:mmind}
For any positive integer $a,b,c$ we have $\tilde{I}(F\odot \langle a,b,c \rangle) = \frac{F\cdot abc}{\max\{a,b,c\}}$.
\end{lemma}

\begin{proof}
Assume without loss of generality that $c \geq a,b$. We have that $\tilde{I}(F\odot \langle a,b,c \rangle) \leq F ab$ since $F\odot \langle a,b,c \rangle$ has only $F ab$ different $x$-variables. In order to show that $\tilde{I}(F\odot \langle a,b,c \rangle) \geq F ab$ and complete the proof, we will show that for every $\delta < 1$, we have $\tilde{I}(F\odot \langle a,b,c \rangle) \geq \delta Fab$. 

Let $n$ be a big enough positive integer so that $\left(\frac34 \right)^{1/n} \geq \delta$. By Lemma \ref{lem:strassenmonomialdegen}, we know that $(F\odot \langle a,b,c \rangle)^{\otimes n}$, which is isomorphic to $F^n\odot \langle a^n, b^n, c^n \rangle$, has a monomial degeneration to an independent tensor of size $F^n \frac34 a^n b^n$. Hence, by Lemma \ref{lem:awmondegtozero}, we have $\tilde{I}((F\odot \langle a,b,c \rangle)^{\otimes n}) \geq F^n \frac34 a^n b^n$, and so by Corollary \ref{cor:mondegtoit}, $\tilde{I}(F\odot \langle a,b,c \rangle) \geq (F^n \frac34 a^n b^n)^{1/n} \geq \delta Fab$, as desired.
\end{proof}

Finally, we can prove the main idea behind our lower bound framework:

\begin{theorem} \label{thm:IROmega}
For any concise tensor $T$, $$\tilde{I}(T) \geq \tilde{R}(T)^{\frac{6}{\omega_g(T)} - 2}.$$
\end{theorem}

\begin{proof} Let $T$ be over $X,Y,Z$.
By Lemma~\ref{lem:omegagdef}, for every $\delta>0$, there are positive integers $n,a,b,c$ such that $T^{\otimes n}$ has a monomial degeneration to $F\odot \langle a, b, c \rangle$, where 

$$abc 
\geq \left( \frac{\tilde{R}(T)^n}{F}\right)^{\frac{3(1-\delta)}{\omega_g(T)}}.$$
Thus, by Lemma \ref{lem:mmind} and Lemma \ref{lem:IAltIB}, we have that $$\tilde{I}(T^{\otimes n}) \geq \tilde{I}(F\odot\langle a,b,c \rangle) = F\cdot\frac{abc}{\max\{a,b,c\}}.$$
Now, by counting variables in $F\odot\langle a,b,c \rangle$, note that $Fab \leq |X^n| \leq \tilde{R}(T)^n$, and hence,
$$c = \frac{abc}{ab} \geq \frac{F abc}{\tilde{R}(T)^n}.$$
Similarly, $a$ and $b$ have the same lower bound. Hence,
$$\max\{ a,b,c \} \leq \frac{abc}{(\min\{a,b,c\})^2} \leq \frac{abc}{(abc F/ \tilde{R}(T)^n)^2} = \frac{\tilde{R}(T)^{2n}}{F^2abc}.$$
We finally get that
$$\tilde{I}(T^{\otimes n}) \geq \frac{Fabc}{\max\{a,b,c\}} \geq \frac{F^3(abc)^2}{\tilde{R}(T)^{2n}}   \geq   F^{3-6(1-\delta)/\omega_g(T)}\tilde{R}(T)^{2(1-\delta)\frac{3n}{\omega_g(T)} - 2n}.$$
Now let $f=\lim_{n\rightarrow \infty} F^{1/n}$. Since $F\geq 1$, we get that $f\geq 1$. We obtain:
$$\tilde{I}(T) \geq f^{3-6(1-\delta)/\omega_g(T)} \tilde{R}(T)^{2(1-\delta)\frac{3}{\omega_g(T)} - 2}\geq \tilde{R}(T)^{2(1-\delta)\frac{3}{\omega_g(T)} - 2},$$
where the last inequality holds since $f\geq 1$ and  $3-6(1-\delta)/\omega_g(T)\geq 0$. The result follows since the inequality above holds for all $\delta>0$.
\end{proof}

\begin{corollary} \label{cor:omegaandi}
For any tensor $T$, if $\omega_g(T)=2$, then $\tilde{I}(T) = \tilde{R}(T)$. Moreover, for every constant $s<1$, there is a constant $w>2$ such that every tensor $T$ with $\tilde{I}(T) \leq \tilde{R}(T)^s$ must have $\omega_g(T) \geq w$.
\end{corollary}

\section{Partitioning Tools for proving lower bounds} \label{sec:partitioning}

The goal of this section is to show some `local' properties of tensors $T$ which imply upper bounds on $\tilde{I}(T)$ (and hence, they will be ultimately used to prove lower bounds on $\omega_g(T)$). The general idea is that we will be finding partitions $T = A + B$ of our tensors, such that at least one of $\tilde{I}(A)$ and $\tilde{I}(B)$ is low, and using this to show that $\tilde{I}(T)$ is itself low. If $\tilde{I}$ were additive, i.e. if it were the case that $\tilde{I}(T) = \tilde{I}(A) + \tilde{I}(B)$ for any partition $T = A+B$, then this would be relatively straightforward. Unfortunately, $\tilde{I}$ is not additive in general, and even in many natural situations:

\begin{example}
Let $q$ be any positive integer, and define the tensors $T_1 := \sum_{i=0}^q x_0 y_i z_i$, $T_2 := \sum_{i=1}^{q+1} x_i y_0 z_i$, and $T_3 := \sum_{i=1}^{q+1} x_i y_i z_{q+1}$. We can see that $T_1$ has only one $x$-variable, $T_2$ has only one $y$-variable, and $T_3$ has only one $z$-variable, and so $\tilde{I}(T_1) = \tilde{I}(T_2) = \tilde{I}(T_3) = 1$. However, $T_1 + T_2 + T_3 = CW_q$, so the three tensors give a partition of the Coppersmith-Winograd tensor! Combining Lemma \ref{lem:mmind} with the fact that $CW_q^{\otimes n}$ is known to zero out into fairly large matrix multiplication tensors for a large enough constant $n$, we see that $\tilde{I}(T_1 + T_2 + T_3)$ can grow unboundedly large as we increase $q$ (in particular, we will see in Theorem~\ref{thm:CWIlb} that $\tilde{I}(T_1 + T_2 + T_3) \geq (q+2)^{2/3}$). We can similarly see that $\tilde{I}(T_1 \otimes T_2 \otimes T_3)$ grows unboundedly with $q$, and so $\tilde{I}$ is not \emph{multiplicative} either.
\end{example}

Throughout this section, we will nonetheless describe a number of general situations where, if $T$ is partitioned into $T = A+B$, then bounds on $\tilde{I}(A)$ and $\tilde{I}(B)$ are sufficient to give bounds on $\tilde{I}(T)$.

We begin with some useful terminology and notation about partitioning tensors. Let $D$ be a sub-tensor of a tensor $T$, that is, it is obtained by removing triples from the support of $T$. If $T$ is over variable sets $X=\{x_1,\ldots,x_a\},Y=\{y_1,\ldots,y_b\},Z=\{z_1,\ldots,z_c\}$, then $T^{\otimes n}$, and hence $D^{\otimes n}$, is over variable sets $\bar{X},\bar{Y},\bar{Z}$, where the variables in $\bar{X}$ are indexed by $n$-length sequences over $[a]$, the variables in $\bar{Y}$ are indexed by $n$-length sequences over $[b]$, the variables in $\bar{Z}$ are indexed by $n$-length sequences over $[c]$. 

\begin{definition} Let $T$ be a partitioned tensor $T=\sum_i P_i$, and let $D$ be a sub-tensor of $T^{\otimes n}$.
Consider some $j\in \{1,\ldots, n\}$.
We say that 
$D$ has an entry of $P_i$ in the $j$th coordinate if there is a triple $(\alpha,\beta,\gamma)$
in the support of $D$ for which $(\alpha_j,\beta_j,\gamma_j)$ is in the support of $P_i$. 
\end{definition}

Since the $P_i$ partition the triples in the support of $T$, this is well-defined.

We begin with our first partitioning tool, which we interpret after the Theorem statement.

\begin{theorem} \label{thm:removeanx}
Suppose $T$ is a tensor over $X,Y,Z$ with $|X|=q$, and $x_1 \in X$ is any $x$-variable such that $x_1$ is in at most $q$ terms in $T$. Let $B := T|_{X \setminus \{x_1\}}$ be the tensor over $X \setminus \{x_1\},Y,Z$ from zeroing out $x_1$ in $T$, and suppose that $c := \tilde{I}(B)$ satisfies $$c \leq \frac{q-1}{q^{1/(q-1)}}.$$ Then, $$\tilde{I}(T) \leq \left( \frac{q-1}{1-p} \right)^{1-p} \cdot \frac{1}{p^p},$$ where $p \in [0,1]$ is given by $$p := \frac{\log\left(\frac{q-1}{c}\right)}{\log\left(q\right) + \log\left(\frac{q-1}{c}\right)}.$$
\end{theorem}

\begin{remark}
Before we prove Theorem \ref{thm:removeanx}, let us briefly interpret its meaning. Since $B$ has only $q-1$ different $x$-variables, we know that $\tilde{I}(B) \leq q-1$. The theorem tells us that if, in fact, $\tilde{I}(B)$ is mildly smaller than this, then regardless of what terms in $T$ involve $x_1$, we still get a nontrivial upper bound on $\tilde{I}(T)$. One can verify that $p = 1/q$ when $c = (q-1)/q^{1/(q-1)}$, and for every $c$ less than this, $p>1/q$, which gives a resulting bound on $\tilde{I}(T)$ which is strictly less than $q$.
\end{remark}

\begin{proof}[Proof of Theorem \ref{thm:removeanx}]

Let $A := T|_{x_1}$ be the tensor over $\{ x_1 \}, Y, Z$ from zeroing out all the $x$-variables other than $x_1$ in $T$. Hence, $T = A+B$ is a partition of $T$. Moreover, since $A$ only has a single $x$-variable, we have $\tilde{I}(A)=1$.

For any positive integer $n$, let $g_n$ be the largest integer such that $T^{\otimes n}$ has a zeroing out into an independent tensor $D_n$ of size $|D_n| = g_n$.

Set $T' = T^{\otimes n}$ and $D' = D_n$, and then for $j$ from $1$ to $n$ do the following process:

Currently $T' = Q_1 \otimes Q_2 \otimes \cdots \otimes Q_{j-1} \otimes T^{n-j+1}$, and $|D'| \geq q_1 q_2 \cdots q_{j-1} \cdot |D_n|$, and moreover, $D'$ is a zeroing out of $T'$. Since $T=A+B$ is a partitioning of $T$, it must be the case that either at least a $p$ fraction of the independent triples in $D'$ have an entry of $A$ in their $j$th coordinate, or else at least a $1-p$ fraction of the independent triples in $D'$ have an entry of $B$ in their $j$th coordinate. In the former case, set $Q_j = A$ and $q_j = p$, and in the latter case, set $Q_j = B$ and $q_j = 1-p$. Recall that there is a zeroing out $z$ such that $z(T') = D'$. Now, replace the $j$th tensor in the product defining $T'$ by $Q_j$, i.e. set $T' = Q_1 \otimes Q_2 \otimes \cdots \otimes Q_{j} \otimes T^{n-j}$. By our choice of $Q_j$, we know that if we apply the same zeroing out $z$ to the new $T'$, we get at least a $q_j$ fraction of the number of independent triples we had before, i.e. $|z(T')| \geq q_j |D'|$. Let $D'$ be this new independent tensor $z(T')$.

Once we have done this for all $j$, we are left with a tensor $\bigotimes_{j=1}^n Q_j$ which has a zeroing out into $|D| \cdot \prod_{j=1}^n q_j$ independent triples. Suppose that we picked $Q_j = A$ in $k$ of the steps, and hence picked $Q_j = B$ in the remaining $n-k$ of the steps. Hence, we have a zeroing out of $A^{\otimes k} \otimes B^{\otimes n-k}$ into $t := g_n \cdot p^k \cdot (1-p)^{n-k}$ independent triples.

We will now give two different upper bounds on $t$. First, we will count $x$-variables. Since $A$ has only one $x$-variable, and $B$ has at most $q-1$ different $x$-variables, our tensor $A^{\otimes k} \otimes B^{\otimes n-k}$ must have at most $(q-1)^{n-k}$ different $x$-variables. Hence, $t \leq (q-1)^{n-k}$.

Second, we will use our bound $c = \tilde{I}(B)$. This implies that $B^{\otimes (n-k)}$ can zero out into at most $c^{n-k}$ independent triples. Hence, since $A$ has at most $q$ terms, and so $A^{\otimes k}$ has at most $q^k$ terms, we know that $A^{\otimes k} \otimes B^{\otimes (n-k)}$ can zero out into at most $q^k c^{n-k}$ independent triples. In other words, $t \leq q^k c^{n-k}$.

Combining the two upper bounds, we see that
\begin{align*}
   t \leq \min \{ (q-1)^{n-k}, q^k c^{n-k} \}.
\end{align*}
Hence,
\begin{align}
    \label{boundg} g_n = \frac{t}{p^k (1-p)^{n-k}} \leq \min \left\{ \frac{1}{p^k}\left(\frac{q-1}{1-p}\right)^{n-k}, \left( \frac{q}{p} \right)^k \left(\frac{c}{1-p}\right)^{n-k} \right\}.
\end{align}
We can see (by setting the two terms equal and solving for $k$) that the right-hand side of (\ref{boundg}) is maximized when $k = pn$. We therefore get a bound independent of $k$ which must hold no matter what $k$ ends up being:

\begin{align*}
    g_n \leq \frac{1}{p^{pn}}\left(\frac{q-1}{1-p}\right)^{(1-p)n}.
\end{align*}

Thus, $$I(T^{\otimes n}) \leq \left( \frac{1}{p^{p}}\left(\frac{q-1}{1-p}\right)^{1-p} \right)^n,$$
and since this holds for all positive integers $n$, it implies our desired bound.
\end{proof}

We next move on to our second tool. We show that if a tensor $T$ has a large asymptotic independence number, then there must be a way to define a probability distribution on the terms of $T$ such that each variable is assigned approximately the same probability mass.

\begin{theorem} \label{thm:probs}
Suppose $q \geq 2$ is an integer, and $T$ is a tensor over $X,Y,Z$ with $|X| = |Y| = |Z| = q$, and $\delta \geq 0$ is such that $\tilde{I}(T) = q^{1-\delta}$. Then, for every $\kappa > 0$, there is a map $p : X \otimes Y \otimes Z \to [0,1]$ such that:
\begin{itemize}
    \item $\sum_{x_i y_j z_k \in T} p(x_i y_j z_k) = 1$, and
    \item For each fixed $i$, fixed $j$, or fixed $k$, $\sum_{x_i y_j z_k \in T} p(x_i y_j z_k) \geq \frac{1}{q} - \sqrt{(\delta + \kappa) \ln(q)}.$
\end{itemize}
\end{theorem}

Before proving Theorem \ref{thm:probs}, we first prove a key Lemma:

\begin{lemma} \label{lem:allsame}
For any integers $n \geq 1$ and $q \geq 2$, any real $\delta \geq 0$, and any tensor $T$ over $X,Y,Z$ with $|X|=q$ and $x_1 \in X$, suppose $T^{\otimes n}$ has a zeroing out into an independent tensor $D$ of size $|D| = q^{(1-\delta)n}$. Let $S_X \subseteq X^n$ be the set of all $x$-variables used in terms in $D$, and let $\eps = \sqrt{\delta \ln(q)}$. Then, at least $q^{(1-\delta)n} - q^{(1-2\delta)n}$ of the elements $x \in S_X$ have $x_1$ appear in between $(1/q - \eps)n$ and $(1/q + \eps)n$ of the entries of $x$.
\end{lemma}

\begin{proof}
Notice that the number of different $n$-tuples of variables of $X$ which contain $x_1$ exactly $i$ times is $\binom{n}{i} \cdot (q-1)^{n-i}$. Hence, the number of elements $x \in X^n$ which do not have $x_1$ appear in between $\frac{1-\eps}{q}n$ and $\frac{1+\eps}{q}n$ of the entries of $x$ is
\begin{align}\label{cherbound}\sum_{i =0}^{\frac{1-\eps}{q}n} \binom{n}{i} (q-1)^{n-i}    +   \sum_{i =\frac{1+\eps}{q}n}^{n} \binom{n}{i} (q-1)^{n-i} .\end{align}

We will bound the sum (\ref{cherbound}) using Hoeffding's inequality\footnote{Hoeffding's inequality states that if $X_1, \ldots, X_n$ are independent random variables taking on values in $[0,1]$, then for any $t \in [0,1]$, we have $\Pr[\sum_{i=1}^n X_i - \E[\sum_{i=1}^n X_i] \geq tn] \leq e^{-2nt^2}$.}. Let $A_1, \ldots, A_n$ be $n$ independent random variables taking on the value $1$ with probability $1/q$ and $0$ otherwise, and let $A = \sum_{i=1}^n A_i$. We can see that (\ref{cherbound}) is equal to $q^n \cdot \Pr[|A - n/q| \geq \eps n]$. By Hoeffding's inequality, if we pick $\eps = \sqrt{\delta \ln(q)}$, then $\Pr[|A - n/q| \geq \eps n] \leq q^{-2 \delta n}$. Thus, (\ref{cherbound}) is at most $q^n \cdot q^{-2 \delta n} = q^{(1-2\delta)n}$, and the result follows.
\end{proof}

\begin{proof}[Proof of Theorem \ref{thm:probs}]
Suppose $\tilde{I}(T) = q^{1-\delta}$, and let $\delta' = \kappa/2 > 0$. Thus, there is a positive integer $N$ such that for all $n \geq N$, the tensor $T^{\otimes n}$ has a zeroing out into an independent tensor $D$ of size $|D| = q^{n(1-\delta-\delta')}$.

Each term in $T^{\otimes n}$, and hence in $D$, corresponds to an $n$-tuple of terms from $T$. We thus define a probability distribution $p : X \otimes Y \otimes Z \to [0,1]$ as follows: draw a uniformly random $\alpha \in \{1, \ldots, n\}$, then draw a uniformly random one of the $|D|$ independent triples from $D$ and return its entry in the $\alpha$th coordinate. Since this random process always returns a term from $T$, we have $\sum_{x_i y_j z_k \in T} p(x_i y_j z_k) = 1$.

Now, pick any fixed $i$ and consider the sum $p(x_i) := \sum_{x_i y_j z_k \in T} p(x_i y_j z_k)$. Let $S_X \subseteq X^n$ be the set of all $X$-variables used in terms of $D$, so $|S_X|=|D|=q^{n(1-\delta-\delta')}$. Then, $p(x_i)$ can be alternatively characterized as the probability, upon drawing a random $\alpha \in \{1,2,\ldots,n\}$ and random $X_s \in S_X$, that the $\alpha$th coordinate of $X_s$ is $x_i$. By Lemma~\ref{lem:allsame}, setting $\eps = \sqrt{(\delta+\delta')\ln(q)}$, we know that for all but $q^{n(1-2\delta-2\delta')}$ of the $X_s \in S_X$, the variable $x_i$ appears in between $(1/q - \eps)n$ and $(1/q + \eps)n$ of the entries of $X_s$. Hence,

$$p(x_i) \geq \frac{(1/q - \eps)n \cdot (q^{n(1-\delta-\delta')} - q^{n(1-2\delta-2\delta')})}{n \cdot q^{n(1-\delta-\delta')}} = (1/q - \eps)(1 - q^{-n(\delta + \delta')}).$$

By a symmetric argument, this same lower bound holds for all of the variables in $X, Y,$ and $Z$. Notice that as $n \to \infty$, the lower bound approaches $(1/q - \eps)$, and $(1/q - \eps) > 1/q - \sqrt{(\delta + \kappa)\ln(q)}$. We can thus pick a sufficiently large $n$ so that the resulting probability distribution has all the desired properties.
\end{proof}

For one simple but interesting Corollary, we will show that in any tensor $T$ which has two `corner terms' (see the Corollary statement for the precise meaning; we will see later that many important tensors have these corner terms), then no matter what the remainder of $T$ looks like, $T$ still does not have too large of an asymptotic independence number.

\begin{corollary} \label{cor:corners}
Suppose $q \geq 2$ is an integer, and $T$ is a tensor over $X,Y,Z$ with $|X| = |Y| = |Z| = q$, such that $x_1, x_q \in X$, $y_1, y_q \in Y$, $z_1 \in Z$, and $T$ contains the triples $x_q y_1 z_1$ and $x_1 y_q z_1$, and neither $x_q$ nor $y_q$ appears in any other triples in $T$. Then, there is a constant $c_q < q$ depending only on $q$ such that $\tilde{I}(T) \leq c_q$.
\end{corollary}

\begin{proof}
Suppose $\tilde{I}(T) = q^{1-\delta}$, and for any $\kappa>0$, let $p$ be the probability distribution on the terms of $T$ which is guaranteed by Theorem~\ref{thm:probs}. For any fixed $i$, define $p(x_i) := \sum_{x_i y_j z_k \in T} p(x_i y_j z_k)$, and define $p(y_j)$ and $p(z_k)$ similarly. Since $x_q y_1 z_1$ and $x_1 y_q z_1$ are the only terms containing $x_q$ or $y_1$, and they each contain $z_1$, it follows that $p(z_1) \geq p(x_q) + p(y_q)$. 

However, we know that $p(x_q), p(y_q) \geq 1/q  - \sqrt{(\delta+\kappa)\ln(q)}$, and so $p(z_1)  \geq 2/q  - 2\sqrt{(\delta+\kappa)\ln(q)}$. Similarly, applying the lower bound on $p(z_i)$ for all $i>1$, we see that $p(z_1) \leq 1 - (q-1)(1/q  - \sqrt{(\delta+\kappa)\ln(q)})$. Combining the two bounds shows that
$$2/q  - 2\sqrt{(\delta+\kappa)\ln(q)} \leq 1 - (q-1)(1/q  - \sqrt{(\delta+\kappa)\ln(q)}) ,$$
and hence, rearranging,
$$1/q   \leq (q+1) \sqrt{(\delta+\kappa)\ln(q)}$$
$$\left( \frac{1}{q(q+1)\sqrt{\ln(q)}} \right)^2 - \kappa \leq \delta.$$
Since this holds for all $\kappa>0$, it implies a lower bound on $\delta$ in terms of $q$ as desired.
\end{proof}

Finally, we move on to our third partitioning tool. This third tool generalizes the fact that if $T$ is a tensor over $X,Y,Z$, then $\tilde{I}(T) \leq \min\{|X|, |Y|, |Z| \}$, i.e. $\tilde{I}(T)$ must be small if $T$ does not have many of one type of variable. We will show that, even if $T$ can be partitioned into tensors which each do not have many of one type of variable, then $\tilde{I}(T)$ must still be small. We will formalize this idea by introducing the notion of the measure of a tensor:

\begin{definition}
Let $T$ be a tensor over $X, Y, Z$.
We say that $X'\subseteq X$, $Y'\subseteq Y$, $Z'\subseteq Z$ are {\em minimal} for $T$ if $X'$ is the minimal (by inclusion) subset of $X$ such that for each $x_i\in X\setminus X'$, for all $j,k$, $T_{i,j,k}=0$, and similarly, $Y'$ is the minimal subset of $Y$ such that for each $y_j\in Y\setminus Y'$, for all $i,k$, $T_{i,j,k}=0$ and $Z'$ is the minimal subset of $Z$ such that for each $z_k\in Z\setminus Z'$, for all $i,j$, $T_{i,j,k}=0$.

If $T$ is a tensor, then the \emph{measure of $T$}, denoted $\mu(T)$, is given by $\mu(T) := |X'| \cdot |Y'| \cdot |Z'|$, where $X', Y', Z'$ are minimal for $T$.
\end{definition}

\begin{claim} \label{lem:measureone}
For any tensor $T$, we have $\tilde{I}(T) \leq \mu(T)^{1/3}$.
\end{claim}

\begin{proof}
Suppose $X, Y, Z$ are minimal for $T$. Hence, $$\tilde{I}(T) \leq \min(|X|, |Y|, |Z|) \leq (|X| \cdot |Y| \cdot |Z|)^{1/3} = \mu(T)^{1/3}.$$
\end{proof}

For our main tool, we can generalize this to partitioned tensors:

\begin{theorem} \label{thm:measures}
Suppose $T$ is a tensor which is partitioned into $k$ parts $T = P_1 + P_2 + \cdots + P_k$ for any positive integer $k$. Then, $\tilde{I}(T) \leq \sum_{i=1}^k (\mu(P_i))^{1/3}$.
\end{theorem}

\begin{proof}
Let $s := \sum_{i=1}^k (\mu(P_i))^{1/3}$, and for each $i \in \{1,2,\ldots,k\}$, let $p_i := (\mu(P_i))^{1/3} / s$, so that $p_i \in [0,1]$ and $\sum_{i=1}^k p_i = 1$. For any positive integer $n$, let $D_n$ be the biggest independent tensor which can result from a zeroing out of $T^{\otimes n}$, and let $z$ be the zeroing out from $T^{\otimes n}$ to $D_n$.

Set $T' = T^{\otimes n}$, and $D' = D_n$, and then for $j$ from $1$ to $n$ do the following process:

Currently $T' = Q_1 \otimes Q_2 \otimes \cdots \otimes Q_{j-1} \otimes T^{n-j+1}$, and $|D'| \geq q_1 q_2 \cdots q_{j-1} \cdot |D_n|$, and moreover, $D'$ is a zeroing out of $T'$. Pick an $i$ such that at least a $p_i$ fraction of the independent triples in $D'$ have an entry of $P_i$ in their $j$th coordinate; since $\sum_\ell p_\ell=1$, such an $i$ exists. Set $Q_j = P_i$ and $q_j = p_i$. Recall that there is a zeroing out $z$ such that $z(T') = D'$. Now, replace the $j$th tensor in the product defining $T'$ by $Q_j$, i.e. set $T' = Q_1 \otimes Q_2 \otimes \cdots \otimes Q_{j} \otimes T^{n-j}$. By our choice of $Q_j$, we know that if we apply the same zeroing out $z$ to the new $T'$, we get at least a $q_j$ fraction of the number of independent triples we had before, i.e. $|z(T')| \geq q_j |D'|$. Let $D'$ be this new independent tensor $z(T')$.

Once we have done this for all $j$, we are left with a tensor $\bigotimes_{j=1}^n Q_j$ which has a zeroing out into $|D_n| \cdot \prod_{j=1}^n q_j$ independent triples. Note that measure is multiplicative, and so in particular, $\mu(\bigotimes_{j=1}^n Q_j) = \prod_{j=1}^n \mu(Q_j)$. Hence, by Claim~\ref{lem:measureone}, $$\tilde{I}(\bigotimes_{j=1}^n Q_j) \leq \prod_{j=1}^n \mu(Q_j)^{1/3} = \prod_{j=1}^n (s \cdot q_j) = s^n \cdot \prod_{j=1}^n q_j.$$
Since $D'$ is a zeroing out of $\bigotimes_{j=1}^n Q_j$, it follows that $|D'| \leq s^n \cdot \prod_{j=1}^n q_j$. But, $|D'| \geq |D_n| \cdot \prod_{j=1}^n q_j$. Combining the two, we get that $|D_n| \leq s^n$, as desired.
\end{proof}

\section{Lower Bounds for Group Tensors} \label{sec:grouptools}
In contrast with the previous section, in this section we will show a `global' property of tensors $T$ which imply upper bounds on $\tilde{I}(T)$ (and hence lower bounds on $\omega_g(T)$). In particular, we will see that if $T$ is the group tensor of any finite group $G$, or a monomial degeneration of any such group tensor with the same measure, then $\tilde{I}(T) < \tilde{R}(T)$ and so $\omega_g(T) > 2$. We begin with the main connection between group tensors and independent tensors; this was essentially proved in \cite[Theorem~6.1]{almanitcs}, but we reprove it here for completeness:

\begin{lemma} \label{lem:indtotcsfs}
For any finite group $G$, if $T_G$ has a zeroing out into an independent tensor $D$, then $G$ has a tri-colored sum-free set of size $|D|$.
\end{lemma}

\begin{proof}
Let $S := \{ (a,b,c) \in G^3 \mid x_a y_b z_c \in D\}$. We will show that $S$ is a tri-colored sum-free set in $G$. First, recall that every $x_a y_b z_c \in T_G$ has $ab=c$, and $D \subseteq T_G$, and so every $(a,b,c) \in S$ has $ab=c$ as well. Second, assume to the contrary that there are $(a_1, b_1, c_1), (a_2, b_2, c_2), (a_3, b_3, c_3) \in S$, not all the same triple, such that $a_1 b_2 = c_3$. This means that none of $x_{a_1}, y_{b_2}$, or $z_{c_3}$ were zeroed out to get from $T_G$ to $D$. But, $x_{a_1} y_{b_2} z_{c_3} \in T_G$, and so we must have $x_{a_1} y_{b_2} z_{c_3} \in D$. Since $D$ is independent, this means that $x_{a_1} y_{b_1} z_{c_1}, x_{a_2} y_{b_2} z_{c_2},$ and $x_{a_3} y_{b_3} z_{c_3}$ must all be the same triple, contradicting how we picked them.
\end{proof}

We can use this to give our main group-theoretic tool for proving lower bounds on $\omega_g$:

\begin{corollary} \label{cor:groupdegen}
For any tensor $T$ and any nontrivial finite group $G$ such that there is a monomial degeneration from $T_G$ into $T$, we have $\tilde{I}(T) < |G|$.
\end{corollary}

\begin{proof}
Since $T$ is a monomial degeneration of $T_G$, by Lemma~\ref{lem:IAltIB} we have $\tilde{I}(T) \leq \tilde{I}(T_G)$. Letting $\delta < 1$ be the constant from Theorem \ref{thm:multmatch} for $G$, we know that for any positive integer $n$, any tri-colored sum-free set in $G^n$ has size at most $(\delta |G|)^n$. Hence, by Lemma \ref{lem:indtotcsfs}, we have $I(T_G^{\otimes n}) \leq (\delta |G|)^n$. It follows by definition that $\tilde{I}(T_G) \leq \delta |G| < |G|$, as desired.
\end{proof}

\begin{theorem} \label{thm:groupomeganot2}
For any finite group $G$, we have $\omega_g(T_G) > 2$.
\end{theorem}

\begin{proof}
There is trivially a monomial degeneration from $T_G$ to itself, so this follows immediately from Corollary \ref{cor:groupdegen} and Corollary \ref{cor:omegaandi}.
\end{proof}

\begin{remark}
This shows that no fixed group tensor $T_G$ can be used to show $\omega=2$ using the Galactic Method. That said, it does not rule out showing $\omega=2$ by using a \emph{sequence} $G_1, G_2, \ldots$ of groups such that $\lim_{i \to \infty} \omega_g(T_{G_i}) = 2$; such a sequence could still exist. Prior work has already made a similar remark for showing $\omega=2$ by finding large `simultaneous triple product property' constructions in $G$ via the Group Theoretic Method, and some natural sequences of groups have already been ruled out \cite{blasiak2017groups}. Although this method is less general than the Galactic Method, their proofs can be combined with the above to rule out these sequences of groups in the Galactic Method as well.
\end{remark}

A question arises: does Theorem~\ref{thm:groupomeganot2} already rule out any `natural' tensor from attaining $\omega=2$ using the Galactic Method? In the remainder of this section, we will give a `no' answer to this question, by showing that the Coppersmith-Winograd tensor itself, which has been used to prove all the most recent upper bounds on $\omega$ \cite{coppersmith, stothers, v12, legall}, cannot be ruled out in this way. We will nonetheless rule out the Coppersmith-Winograd tensor later by using the partitioning tools from the previous section. We begin with some useful lemmas about finite abelian groups. 

\begin{lemma} \label{lem:groupsize}
If $G$ is any finite Abelian group, and $g \in G$ is any element other than the identity, then there are at most $|G|/2$ elements $a \in G$ such that $a^2 = g$.
\end{lemma}

\begin{proof}
 For any $g \in G$ with $g \neq 1$, let $S_g := \{ a \in G \mid a^2 = g\}$ and $S_1 := \{ a \in G \mid a^2 = 1\}$, and suppose that $S_g$ is nonempty. Pick any element $\sqrt{g} \in S_g$. There is hence a bijection $b : S_1 \to S_g$ given by $b(a) = a\sqrt{g}$. Since $S_1$ and $S_g$ are disjoint subsets of $G$ with $|S_1|=|S_g|$, we must have $|S_g| \leq |G|/2$ as desired.
\end{proof}

\begin{lemma} \label{lem:CWnonabelain}
For any positive integer $q$, $CW_q$ is not a sub-tensor of $T_G$ for any abelian group $G$ of order $|G|<2q$.
\end{lemma}

\begin{proof}
Recall that (under a slight change $z_0\longleftrightarrow z_{q+1}$):
$$CW_q = x_0 y_0 z_0 + x_{q+1} y_0 z_{q+1} + x_0 y_{q+1} z_{q+1} + \sum_{i=1}^q (x_i y_0 z_i + x_0 y_i z_i + x_i y_i z_{q+1}).$$

Assume to the contrary that $CW_q$ is a sub-tensor of $T_G$ for some abelian group $G$ of order $|G|<2q$. Let $X,Y,Z$ be the sets of variables of $CW_q$, and let $\bar{X} = \{\bar{x}_g\}_{g \in G}$, $\bar{Y} = \{\bar{y}_g\}_{g \in G}$, and $\bar{Z} = \{\bar{z}_g\}_{g \in G}$ be the sets of variables of $T_G$. That means there are injections $a, b, c : \{0,1,\ldots,q+1\} \to G$ such that if $x_i y_j z_k \in CW_q$, then $\bar{x}_{a(i)} \bar{y}_{b(j)} \bar{z}_{c(k)} \in T_G$. Since $G$ is abelian, we can assume without loss of generality that $a(0) = b(0) = c(0) = 1$, the identity in $G$, since otherwise, replacing $a(i)$ with $a(i) a(0)^{-1}$ for all $i$, replacing $b(j)$ with $b(j) b(0)^{-1}$ for all $j$, and replacing $c(k)$ with $c(k) c(0)^{-1}$ for all $k$, does not change the desired properties of $a,b,c$.

Now, note that since for all $i \in \{1,2,\ldots,q+1\}$, we have $x_i y_0 z_i \in CW_q$, this means that we must have $a(i) = a(i) b(0) = c(i)$ for all such $i$ (by definition of $T_G$). Similarly, since $x_0 y_i z_i \in CW_q$, we must have $b(i) = c(i)$ for all $i \in \{1,2,\ldots,q+1\}$. In fact, $a,b,$ and $c$ are all the same function.

Finally, let $g = c(q+1) \in G$. We have that $g \neq 1$ since $c(0)=1$ and $c$ is an injective function. Meanwhile, for all $i \in \{1,2,\ldots,q\}$, we have that $x_i y_i z_{q+1} \in CW_q$, and so $a(i)^2 = a(i)b(i) = c(q+1)=g$. In other words, for all $q$ different values of $a(i) \in G$ for $i \in \{1,2,\ldots,q+1\}$, we have $a(i)^2 = g$. It follows from Lemma \ref{lem:groupsize} that $|G| \geq 2q$, as desired.
\end{proof}

\begin{remark}
There are many values of $q$ for which Lemma~\ref{lem:groupsize} is tight. For instance, if $C_\ell$ denotes the cyclic group of order $\ell$, then for any nonnegative integer $k$, the group $C_2^k \times C_4$, which has order $2^{k+2}$, contains $CW_{2^{k+1}}$ as a sub-tensor of its group tensor, and even as a monomial degeneration!.
\end{remark}

\begin{theorem}
$CW_q$ is not a sub-tensor of $T_G$ for any group $G$ of order $|G| = q+2$ for $q=3,4,5,6,7,8$, or $9$.
\end{theorem}

\begin{proof}
For $q=3,5,7,9$, the result follows from Lemma~\ref{lem:CWnonabelain} since for those $q$, there is no non-abelian group of order $q+2$, and we have $2q > q+2$. For $q=4,6,8$, there are four different nonabelian groups to check in total, but an argument similar to the proof of Lemma~\ref{lem:CWnonabelain}, or simply a small brute-force search, shows that none of them contradicts the Theorem statement, as desired.
\end{proof}

\begin{remark}
It is not hard to see that $CW_q$ is a sub-tensor (and even a monomial degeneration!) of $T_{q+2}$ for $q=1$ and $q=2$.
\end{remark}

\section{Applications of our Lower Bound Techniques}

In this section, we use the lower bounding techniques that we have developed throughout the paper for a number of applications to tensors of interest.

\subsection{Generalized CW tensors}

We begin by proving our main result:

\begin{theorem} \label{thm:main}
There is a universal constant $c > 2$ such that for any generalized Coppersmith-Winograd tensor $T$ (with any parameter $q$), we have $\omega_g(T) \geq c$.
\end{theorem}
\begin{proof} 
This follows from Lemmas \ref{lem:cwsmall} and \ref{lem:cwbig}, which we state and prove below.
\end{proof}

\begin{lemma} \label{lem:cwsmall}
For every nonnegative integer $q$, there is a constant $c_q > 2$ such that for any generalized Coppersmith-Winograd tensor $T$ with parameter $q$, we have $\omega_g(T) \geq c_q$.
\end{lemma}

\begin{lemma} \label{lem:cwbig}
There is a constant $c' > 2$ and a positive integer $q'$ such that for any integer $q \geq q'$, and any generalized Coppersmith-Winograd tensor $T$ with parameter $q$, we have $\omega_g(T) \geq c'$.
\end{lemma}

\begin{proof}[Proof of Lemma \ref{lem:cwsmall}]
For each $q$, and each generalized Coppersmith-Winograd tensor $T$ with parameter $q$, the tensor $T$ is of the form described by Corollary~\ref{cor:corners}, which says that $\tilde{I}(T) < s_{q+2}$ for some constant $s_{q+2} < q+2$ which depends only on $q$. It then follows from Corollary \ref{cor:omegaandi} that $\omega_g(T) > c_q$ for some constant $c_q > 2$ determined by $s_q$, as desired.
\end{proof}

The proof above of Lemma~\ref{lem:cwsmall} used Corollary~\ref{cor:corners}, which follows from Theorem~\ref{thm:probs}, as its main tool. We will next give two different proofs of Lemma \ref{lem:cwbig}; the first will showcase Theorem~\ref{thm:measures}, and the second will showcase Theorem~\ref{thm:removeanx}. Each of Theorems \ref{thm:removeanx}, \ref{thm:probs}, and \ref{thm:measures} describes a different property of a tensor $T$ which is enough to imply that $\omega_g(T)>2$. Throughout these three proofs, we are showing that the Coppersmith-Winograd tensor has \emph{all three} of these properties!

\begin{proof}[First proof of Lemma \ref{lem:cwbig}]
Suppose $T$ is a generalized Coppersmith-Winograd tensor with parameter $q$. Hence, $T$ can be written as
$$T=x_0 y_0 z_0 + x_0 y_{q+1} z_{q+1} + x_{q+1} y_0 z_{q+1} + \sum_{i=1}^q (x_0 y_i z_i + x_i y_0 z_i + x_i y_{\sigma(i)} z_{q+1}),$$
for some permutation $\sigma$ on $\{1,2,\ldots, q\}$. We partition $T$ into three parts $T_1, T_2, T_3$ as follows:
$$T_1 = \sum_{i=0}^q x_0 y_i z_i,$$
$$T_2 = \sum_{i=1}^{q+1} x_i y_0 z_i,$$
$$T_3 = x_0 y_{q+1} z_{q+1} + \sum_{i=1}^q x_i y_{\sigma(i)} z_{q+1}.$$
Note that $T_1$ has only one $x$-variable, $T_2$ has only one $y$-variable, and $T_3$ has only one $z$-variable. Hence, $\mu(T_1) = \mu(T_2) = \mu(T_3) = q^2$. It follows from Theorem \ref{thm:measures} that $\tilde{I}(T) \leq 3q^{2/3}$. When $q \geq 28$, we have $3 q^{2/3} < q^{0.997}$, and so by Corollary \ref{cor:omegaandi}, there is a fixed constant $c'>2$ independent of $q$ such that $\omega_g(T) \geq c'$, as desired.
\end{proof}

Our second proof will use Theorem~\ref{thm:removeanx} instead of Theorem~\ref{thm:measures} as our primary tool. The arithmetic will be messier, but we will be able to achieve a smaller integer $q'$: $6$ instead of $28$.

\begin{proof}[Second proof of Lemma \ref{lem:cwbig}]
Consider any generalized Coppersmith-Winograd tensor with parameter $q$, which is given by
\[CW^\sigma_q = (x_0y_0z_{q+1}+x_0y_{q+1}z_0+x_{q+1}y_0z_0) + \sum_{i=1}^q (x_i y_{\sigma(i)} z_0 + x_iy_0z_i+x_0y_iz_i).\]
We define two intermediate tensors, $A$ and $B$, given by:
$$A =(x_{q+1}y_0z_0) + \sum_{i=1}^q (x_i y_{\sigma(i)} z_0 + x_iy_0z_i), $$
$$B =\sum_{i=1}^q (x_i y_{\sigma(i)} z_0). $$
Note that $A$ is the tensor over $\{x_1, \ldots, x_{q+1}\}, \{y_0, \ldots, y_q\}, \{z_1, \ldots, z_{q+1}\}$ which results from zeroing out $x_0$ in $CW^\sigma_q$. Moreover, $B$ is the tensor over $\{x_1, \ldots, x_1\}, \{y_1, \ldots, y_1\}, \{z_0\}$ which results from zeroing out $y_0$ in $A$.

We first apply Theorem~\ref{thm:removeanx} to $A$ and $B$. Since $B$ only has a single $z$-variable, we have that $\tilde{I}(B)=1$. Applying the Theorem gives us the bound:
$$\tilde{I}(A) \leq \left( \frac{q^{\log(q+1)} (\log(q^2+q))^{\log(q^2+q)}}{(\log(q+1))^{\log(q+1)} (\log(q))^{\log(q)}} \right)^{\frac{1}{\log(q^2+q)}}.$$
One can confirm that this bound is less than $(q+1)/(q+2)^{1/(q+1)}$ whenever $q \geq 6$. 

In the case of $q=6$, the bound above gives us that $\tilde{I}(A) \leq 5.07905$. We can then apply Theorem~\ref{thm:removeanx} again with $CW^\sigma_6$ and $A$. When doing so, we have $c = 5.07905$, and so we find $p = 0.133648$ and we hence get the bound $\tilde{I}(CW^\sigma_6) \leq 7.9973$. Since this is a constant less than $8$, by Corollary \ref{cor:omegaandi}, we know there is a constant $c_6 > 2$ such that $\omega_g(CW^\sigma_6) > c_6$. Note in particular that $c_6$ is independent of $\sigma$ since we never used what $\sigma$ is. We can then do the same process for any $q>6$ to yield a constant $c_q$, but our bound is improving with $q$, so we will get $c_q \geq c_6$ for all such $q \geq 6$, which completes the proof.
\end{proof}

\subsection{$\tilde{I}$ and tri-colored sum-free set constructions for all finite groups} \label{sec:tcsfslb}

One of the key components to our lower bounding framework is Lemma~\ref{lem:mmind}, in which we showed that matrix multiplication tensors have large asymptotic independence numbers. In this subsection, we will instead use Lemma~\ref{lem:mmind} in a different way: to show that some other tensors of interest also have nontrivially-large asymptotic independence numbers. In particular, we will show this for the group tensor $T_G$ of any finite group $G$, which will imply a nontrivially-large tri-colored sum-free set in $G^n$ for sufficiently large $n$. We start with the main additional idea needed for this application:

\begin{theorem} \label{thm:anygroupcw}
For every finite group $G$ of order $|G|=q$, there is a monomial degeneration of $T_G$ into a tensor $T$ which is a generalized Coppersmith-Winograd tensor with parameter $q-2$.
\end{theorem}

\begin{proof}
Let $1 \in G$ be the identity, and let $g \in G$ be any other element. We define the maps $\alpha : X_G \to \Z$, $\beta : Y_G \to \Z$, and $\gamma : Z_G \to \Z$ which give our monomial degeneration as follows:
\begin{itemize}
    \item $\alpha(x_1) = \beta(y_1) = \gamma(z_1) = 0$,
    \item $\alpha(x_g) = \beta(y_g) = - \gamma(z_g) = 2$, and
    \item $\alpha(x_h) = \beta(y_h) = - \gamma(z_h) = 1$ for all $h \in G \setminus \{1,g\}$.
\end{itemize}

Let $T$ be the monomial degeneration of $T_G$ defined by $\alpha, \beta, \gamma$. Define the permutation $\sigma : G \setminus \{ 1,g \} \to G \setminus \{ 1,g \}$ which sends $h \in G$ to $\sigma(h) := h^{-1} g$. We can see that:

\begin{itemize}
    \item $x_1 y_1 z_1 \in T$ since $\alpha(x_1) = \beta(y_1) = \gamma(z_1) = 0$.
    \item $x_1 y_h z_h \in T$ for all $h \in G \setminus \{1\}$ (including $h=g$), since $\alpha(x_1) =0$ while $\beta(y_h) = - \gamma(z_h) = 1$.
    \item $x_h y_1 z_h \in T$ for all $h \in G \setminus \{1\}$ similarly.
    \item $x_h y_{\sigma(h)} z_g \in T$ for all $h \in G \setminus \{1,g\}$, since $\alpha(x_h) = \beta(y_{\sigma(h)}) = 1$, while $\gamma(z_g) = -2$.
\end{itemize}
Meanwhile,
\begin{itemize}
    \item $x_{h_1} y_{h_2} z_{h_3} \notin T$ for any $h_1, h_2, h_3 \in G \setminus \{1,g\}$ with $h_1 h_2 =  h_3$, since $\alpha(h_1) = \beta(h_2) = 1$ and $\gamma(h_3) = -1$, so the three sum to $1$.
    \item $x_h y_{h^{-1}} z_1 \notin T$ for any $h \in G \setminus \{1,g\}$ since $\alpha(x_h) = \beta(y_{h^{-1}}) = 1$ while $\gamma(z_1) = 0$, so the three sum to $2$.
    \item $x_g y_{h_1} z_{h_2} \notin T$ for any $h_1, h_2 \in G \setminus \{ 1,g \}$ with $g h_1  = h_2$, since $\alpha(x_g) = 2$, $\beta(y_{h_1}) = 1$, and $\gamma(z_{h_2}) = -1$, so the three sum to $2$.
    \item  $x_{h_1} y_{g} z_{h_2} \notin T$ for any $h_1, h_2 \in G \setminus \{ 1,g \}$ with $g h_1  = h_2$ similarly.
    \item $x_g y_{g^{-1}} z_1 \notin T$ since $\alpha(x_g) = 2$, $\beta(y_{g^{-1}}) = 1$, and $\gamma(z_1) = 0$, so the three sum to 3.
    \item $x_{g^{-1}} y_{g} z_1 \notin T$ similarly.
    \item $x_g y_g z_{g^{2}} \notin T$ since $\alpha(x_g) = \beta(y_g) = 2$, and definitely $\gamma(z_{g^{2}}) \geq -2$, so the three sum to at least 2.
\end{itemize}
This covers all the entries of $T_G$, showing that we have defined a valid monomial degeneration to $$T=x_1 y_1 z_1 + x_1 y_g z_g + x_g y_1 z_g + \sum_{h \in G \setminus \{1,g\}} (x_1 y_h z_h + x_h y_1 z_h + x_h y_{\sigma(h)} z_g).$$
This is indeed a generalized Coppersmith-Winograd tensor with parameter $|G \setminus \{1,g\}| = q-2$, as desired.
\end{proof}

\begin{remark} \label{rem:TGomegaUB}
An immediate consequence of this monomial degeneration is that
applying any implementation of the Solar, Galactic or Universal method on $T_G$ for \emph{any} finite group $G$ with $\tilde{R}(T_G) = |G|$ yields the same upper bounds on $\omega$ as the best known analysis of $CW_{|G|-2}$. Picking an appropriate group $G$ where group operations are known to be efficient in practice could help lead to a more practical matrix multiplication algorithm. 
\end{remark}

Next, we will use the fact that matrix multiplication tensors, and hence Coppersmith-Winograd tensors, have large asymptotic independence number, to show that for any finite group $G$, $T_G$ also has a relatively large independence number, and hence that $G^n$ has relatively large tri-colored sum-free sets for large enough $n$.

\begin{theorem} \label{thm:CWIlb}
Define $f : \N \to \R$ by $f(q) = \log_q \left( \frac{4(q+2)^3}{27} \right)$. For every positive integer $q$, and every tensor $T$ which is a generalized Coppersmith-Winograd tensor of parameter $q$, we have $\tilde{I}(T) \geq (q+2)^{2/f(q)}$.
\end{theorem}

\begin{remark} \label{rem:CWIlb}
For $q \geq 3$, we have $f(q) < 3$, and so $\tilde{I}(T) \geq (q+2)^{2/3}$.
\end{remark}

\begin{remark} \label{rem:CWIlb2}
In the proof of Theorem \ref{thm:CWIlb}, we use a simpler lower bound on $\omega_g(CW_q)$ than is known for ease of reading; it is, of course, possible to use the better known upper bounds on $\omega_g(CW_q)$ from \cite{coppersmith, legall} in the proof and improve the result.
\end{remark}

\begin{proof}[Proof of Theorem \ref{thm:CWIlb}]
Define $f : \N \to \R$ by $f(q) = \log_q \left( \frac{4(q+2)^3}{27} \right)$. In \cite[Section 6]{coppersmith}, Coppersmith and Winograd show that $\omega_g(CW_q) \geq f(q)$. Hence, for every $\delta>0$, there is a positive integer $n$ such that $CW_q^{\otimes n}$ has a zeroing out\footnote{In fact, $\omega_g(CW_q) \geq f(q)$ only implies that a matrix multiplication tensor of this volume exists as a zeroing out of $CW_q^{\otimes n}$; the fact that a square one can be found follows from the actual analysis of \cite{coppersmith}, which is symmetric in the three types of variables.} into $\langle t, t, t \rangle$, for $t \geq \tilde{R}(CW_q)^{(1-\delta)n/f(q)} = (q+2)^{(1-\delta)n/f(q)}$.

Because of the blocking used by their application of the Laser method, their bound actually holds for any generalized Coppersmith-Winograd tensor of parameter $q$, meaning that $T^{\otimes n}$ also has a zeroing out into $\langle t,t,t \rangle$. By Lemma \ref{lem:mmind}, we thus have $I(T^{\otimes n}) \geq t^2$, which means as desired that
$$\tilde{I}(T) \geq t^{2/n} \geq (q+2)^{\frac{2(1-\delta)}{f(q)}}.$$
\end{proof}

\begin{theorem}
For every (not necessarily abelian) finite group $G$, there is a constant $c_{|G|} > 2/3$, depending only on $|G|$, such that $\tilde{I}(T_G) \geq |G|^{c_{|G|}}$. In particular, $G^n$ has a tri-colored sum-free set of size at least $|G|^{c_{|G|}n - o(n)}$.
\end{theorem}

\begin{proof}
The only finite groups $G$ of order $|G|<5$ are $C_1, C_2, C_3, C_4$, and $C_2^2$ (where $C_\ell$ denotes the cyclic group of order $\ell$). For each of these groups, the result is shown, eg. by \cite{kleinberg}. For $|G| \geq 5$, we know from Theorem~\ref{thm:anygroupcw} that $T_G$ has a monomial degeneration to a generalized Coppersmith-Winograd tensor of parameter $|G|-2$, and so the result follows by Theorem \ref{thm:CWIlb} (and in particular, Remark \ref{rem:CWIlb2}).
\end{proof}

\subsection{$T_q$ lower triangular}

Recall that, for each positive integer $q$, we defined the tensor $T_q$ (the group tensor of the cyclic group $C_q$) as:
$$T_q=\sum_{i=0}^{q-1}\sum_{j=0}^{q-1} x_i y_j z_{i+j\bmod q}.$$
We can then define the lower triangular version of $T_q$, called $T_q^{lower}$, as:
$$T_q^{lower}=\sum_{i=0}^{q-1}\sum_{j=0}^{q-1-i} x_i y_j z_{i+j}.$$
We clearly have $T_q^{lower} \subseteq T_q$, and in fact, there is a simple monomial degeneration to $T_q$ from $T_q^{lower}$ by picking $a(x_i) = b(x_i)=i$ and $c(z_i) = -i$. $T_q^{lower}$ is a natural tensor in its own right, and the fact that each of its $z$-variables only appears on `diagonals' of $x$ and $y$-variables makes it particularly amenable to analysis using the Laser Method. It is even shown in \cite{almanitcs} that the rotated $CW_q$ tensor has a simple monomial degeneration from $T_q^{lower}$.

Since $T_q$ is the group tensor of $C_q$, we already know from Theorem~\ref{thm:groupomeganot2} that $\omega_g(T_q) > 2$. Moreover, since $T_q^{lower}$ is a  monomial degeneration of $T_q$, we already know that $\omega_g(T_q^{lower})>2$ as well. That said, we can instead give a simpler proof of this fact, which avoids the tri-colored sum-free set framework.

\begin{theorem} \label{thm:tqtriangle}
For each integer $q \geq 2$, there is a constant $c_q > 2$ such that $\omega_g(T_q^{lower}) \geq c_q$.
\end{theorem}

\begin{proof}
For each $q$, the tensor $T_q^{lower}$ is of the form described by Corollary \ref{cor:corners}, which says that $\tilde{I}(T_q^{lower}) < q$. It then follows from Corollary \ref{cor:omegaandi} that $\omega_g(T_q^{lower}) > 2$, as desired.
\end{proof}

\begin{remark}
The main result of \cite{kleinberg} can be interpreted as showing that $\lim_{q \to \infty} \log_q(\tilde{I}(T_q^{lower})) = 1$. Hence, it is impossible to improve Theorem~\ref{thm:tqtriangle} to make $c_q$ be a constant independent of $q$ if our proof only uses a bound on $\tilde{I}(T_q^{lower})$. Interestingly, they also show that for all $q$, $\tilde{I}(T_q) = \tilde{I}(T_q^{lower})$.
\end{remark}

\subsection{Lower Triangular Tensors}

In fact, we can give a strong characterization of lower triangular tensors which are potentially able to prove $\omega=2$ within the Galactic method.

\begin{definition}
For $X = \{x_0, \ldots, x_{q-1}\}$, $Y = \{y_0, \ldots, y_{q-1}\}$ and $Z = \{z_0, \ldots, z_{q-1}\}$, a tensor $T$ over $X,Y,Z$ is \emph{lower triangular} if 
\begin{itemize}
    \item For every $i,j \in \{0,\ldots,q-1\}$, there is at most one $k \in \{0,\ldots,q-1\}$ with $x_i y_j z_k \in T$, and
    \item For every $i,j \in \{0,\ldots,q-1\}$ with $i+j\geq q$, $x_i y_j z_k \notin T$ for any $k \in \{1,\ldots,q\}$.
\end{itemize}
Terms $x_i y_j z_k$ with $i+j=q-1$ are called \emph{diagonal terms}.
\end{definition}

\begin{theorem}
For $X = \{x_0, \ldots, x_{q-1}\}$, $Y = \{y_0, \ldots, y_{q-1}\}$ and $Z = \{z_0, \ldots, z_{q-1}\}$, a lower triangular tensor $T$ over $X,Y,Z$ has $\tilde{I}(T) = q$ if and only if it has $q$ diagonal terms, no two of which share any $z$-variables.
\end{theorem}

\begin{proof}
Consider first any lower diagonal tensor $T$ whose $q$ diagonal terms do not share $z$-variables. There is a simple monomial degeneration from $T$ to only its diagonal terms, given by $a(x_i) = b(y_i) = -i$ and $c(z_i) = q-1$ for all $i$. Since no two of the diagonal terms share $z$-variables, this is a monomial degeneration from $T$ to an independent tensor of size $q$, which implies by Corollary~\ref{lem:IAltIB} that $\tilde{I}(T) = q$.

Second, consider any lower diagonal tensor $T$ with $\tilde{I}(T)=q$. Let $f : \{0,\ldots,q-1\}^2 \to \{0,\ldots,q-1\}$ be the map defining which $z$-variable appears in each term, i.e. such that $x_i y_{j} z_{f(i,j)}$ is the only term containing $x_iy_j$ for each $i,j$ (we assume that such a term exists for each $i,j$; if $T$ is missing any such terms, then the proof is even simpler). By Theorem~\ref{thm:probs}, we know that for every $\kappa>0$, there is a probability distribution $p : X \otimes Y \otimes Z \to [0,1]$ whose support is on the terms of $T$, such that for any fixed $i$, $p(x_i) := \sum_{x_i y_j z_k} p(x_i y_j z_k) \geq 1/q - \kappa$, and similarly for $p(y_j)$ and $p(z_k)$. Summing this lower bound for all $x$-variables other than $x_i$ also shows that $p(x_i) \leq 1/q + (q-1)\kappa$ for each $i$, and similarly for $p(y_j)$ and $p(z_k)$.

We now prove that for each $j \in \{0,\ldots,q-1\}$, we have $p(x_{q-1-j}y_jz_{f(q-1-j,j)}) \geq 1/q - O_q(\kappa)$, where we are thinking of $q$ as a constant, so the $O_q$ hides factors of $q$. We prove this by strong induction on $j$. For the base case, when $j=0$, notice that the term $x_{q-1}y_{0}z_{f(q-1,0)}$ is the only term containing $x_{q-1}$, and so $p(x_{q-1}y_{0}z_{f(q-1,0)}) = p(x_{q-1}) \geq 1/q - \kappa$, as desired.

For the inductive step, note that for each $j' < j$, we have by assumption that $p(x_{q-1-j'}y_{j'}z_{f(q-1-j',j')}) \geq 1/q - O_q(\kappa)$. Therefore, for each such $j'$,

\begin{align*}& p(x_{q-1-j} y_{j'}z_{f(i,j')}) \\ &\leq \sum_{i=0}^{q-2-j'} p(x_i y_{j'}z_{f(i,j')}) \\ &= p(y_{j'}) - p(x_{q-1-j'}y_{j'}z_{f(q-1-j',j')}) \leq (1/q + (q-1)\kappa) - (1/q - O_q(\kappa)) = O_q(\kappa).\end{align*}

It follows as desired that
$$p(x_{q-1-j}y_jz_{f(q-1-j,j)}) = p(x_{q-1-j}) - \sum_{j'=0}^{j-1}p(x_{q-1-j} y_{j'}z_{f(i,j')}) \geq p(x_{q-1-j}) - O_q(\kappa) \geq 1/q - O_q(\kappa).$$

Now, assume to the contrary that there is a $k$ such that $k \neq f(q-1-j,j)$ for any $j$. Thus,
$$p(z_k) \leq 1 - \sum_{j=0}^{q-1} p(x_{q-1-j}y_jz_{f(q-1-j,j)}) \leq 1 - \sum_{j=0}^{q-1} (1/q - O_q(\kappa)) = O_q(\kappa).$$
Picking a sufficiently small $\kappa>0$ contradicts Theorem~\ref{thm:probs}.
\end{proof}

\paragraph{Acknowledgments.} The authors are extremely grateful to JM Landsberg and Joshua A. Grochow for answering their many questions, and to Ryan Williams for his many useful suggestions.

\bibliographystyle{alpha}
\bibliography{papers}

\end{document}